\documentclass[12pt]{article}
\pdfoutput=1
\usepackage{jheppub}
\usepackage{verbatim}
\usepackage{graphics, color}
\usepackage{graphicx}
\usepackage{amsmath}
\usepackage{amssymb}
\usepackage{amsthm}
\usepackage{amsfonts}
\usepackage[utf8]{inputenc}

\newcommand{\be}{\begin{equation}}
\newcommand{\ee}{\end{equation}}
\newcommand{\bfig}{\begin{figure}\begin{center}}
\newcommand{\efig}{\end{center}\end{figure}}
\newcommand{\bi}{\begin{itemize}}
\newcommand{\ei}{\end{itemize}}

\newcommand{\lan}{\langle}
\newcommand{\ran}{\rangle}
\newcommand{\Tr}{\mathrm{Tr}}

\newcommand{\wt}{\widetilde}

\newcommand{\Ll}{\mathcal{L}}
\newcommand{\Hh}{\mathcal{H}}
\newcommand{\HA}{\mathcal{H}_A}
\newcommand{\HAb}{\mathcal{H}_{\ol{A}}}
\newcommand{\Ab}{\ol{A}}
\newcommand{\ab}{\ol{a}}
\newcommand{\Ha}{\mathcal{H}_{a}}
\newcommand{\Hab}{\mathcal{H}_{\ol{a}}}
\newcommand{\Hc}{\mathcal{H}_{code}}
\newcommand{\LH}{\Ll(\Hh)}
\newcommand{\LA}{\Ll_A}
\newcommand{\ol}{\overline}

\newcommand{\EA}{\mathcal{E}_A}
\newcommand{\EAb}{\mathcal{E}_{\Ab}}
\newcommand{\Pc}{P_{code}}
\newtheorem{mydef}{Definition}[section]
\newtheorem{thm}{Theorem}[section]
\newtheorem{prop}{Proposition}[section]
\begin{document}
\title{The Ryu-Takayanagi Formula from Quantum Error Correction}
\author[a]{Daniel Harlow}
\affiliation[a]{Center for the Fundamental Laws of Nature, Harvard University, Cambridge MA, 02138 USA}
\emailAdd{dharlow@physics.harvard.edu}
\abstract{I argue that a version of the quantum-corrected Ryu-Takayanagi formula holds in any quantum error-correcting code.  I present this result as a series of theorems of increasing generality, with the final statement expressed in the language of operator-algebra quantum error correction.  In AdS/CFT this gives a ``purely boundary'' interpretation of the formula.  I also extend a recent theorem, which established entanglement-wedge reconstruction in AdS/CFT, when interpreted as a subsystem code, to the more general, and I argue more physical, case of subalgebra codes.  For completeness, I include a self-contained presentation of the theory of von Neumann algebras on finite-dimensional Hilbert spaces, as well as the algebraic definition of entropy.  The results confirm a close relationship between bulk gauge transformations, edge-modes/soft-hair on black holes, and the Ryu-Takayanagi formula.  They also suggest a new perspective on the homology constraint, which basically is to get rid of it in a way that preserves the validity of the formula, but which removes any tension with the linearity of quantum mechanics.  Moreover they suggest a boundary interpretation of the ``bit threads'' recently introduced by Freedman and Headrick.}
\maketitle

\section{Introduction}
The Anti-de Sitter/Conformal Field Theory (AdS/CFT) correspondence has recently been reinterpreted in the language of quantum error correcting codes \cite{Almheiri:2014lwa,Mintun:2015qda,Pastawski:2015qua,Hayden:2016cfa,Freivogel:2016zsb}.  This language naturally implements several features of the correspondence which were previously somewhat mysterious from the CFT point of view:
\bi

\item\textbf{Radial Commutativity:} To leading order in the gravitational coupling $G$, a local operator in the center of a bulk time-slice should commute with all local operators at the boundary of that slice \cite{Polchinski:1999yd}.  But this seems to be in tension \cite{Almheiri:2014lwa} with the time-slice axiom of local quantum field theory \cite{Streater:1989vi,Haag:1992hx}.
  
\item\textbf{Subregion Duality:} Given a subregion $A$ of a boundary time-slice $\Sigma$, we are able to reconstruct any bulk operator $\phi(x)$ which is in the \textit{causal wedge of A}, denoted $\mathcal{C}_A$ and defined as the intersection of the bulk future and the bulk past of the boundary domain of dependence of $A$, as a CFT operator with support only on $A$ \cite{Hamilton:2006az,Morrison:2014jha,Bousso:2012sj,Czech:2012bh,Bousso:2012mh,Hubeny:2012wa}.  Moreover this reconstruction can be extended  \cite{Czech:2012bh,Wall:2012uf,Headrick:2014cta,Jafferis:2015del,Dong:2016eik} into the larger \textit{entanglement wedge of A}, denoted $\mathcal{E}_A$ and defined as the bulk domain of dependence of any bulk achronal surface $\Xi$ whose only boundaries are $A$ and the Hubeny/Rangamani/Takayanagi (HRT) surface $\gamma_A$ associated to $A$ \cite{Hubeny:2007xt}.  Subregion duality implies a remarkable redundancy in the CFT representation of bulk operators, which is illustrated in figure \ref{subregions}.
\bfig
\includegraphics[height=4.5cm]{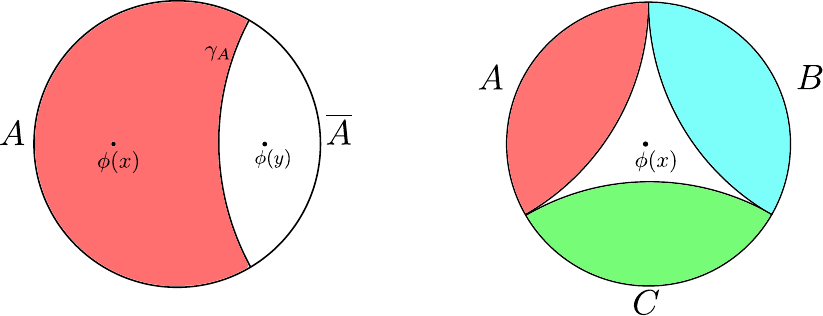}
\caption{Subregion duality in AdS/CFT.  In the left diagram I've shaded the intersection of the entanglement wedge $\EA$ of a boundary subregion $A$ with a bulk time-slice.  The operator $\phi(x)$ is in $\EA$, and thus has a representation in the CFT on $A$. The operator $\phi(y)$ is in $\mathcal{E}_{\ol{A}}$, and thus has a representation on $\ol{A}$.  In the right diagram, we have a situation where $\phi(x)$ has no representation on $A$, $B$, or $C$, but does have a representation on $AB$, $BC$, or $AC$.}\label{subregions}
\efig
\item\textbf{Ryu-Takayanagi Formula:} Given a CFT state $\rho$, we can define a boundary state $\rho_A$ on any boundary subregion $A$.  If $\rho$ is ``appropriate'' then the von Neumann entropy of $\rho_A$ is given by \cite{Ryu:2006bv}, \cite{Hubeny:2007xt,Lewkowycz:2013nqa,Barrella:2013wja,Faulkner:2013ana}
\be\label{flm}
S(\rho_A)=\Tr \left(\rho \LA \right)+S_{bulk}(\rho_{\EA}).
\ee
Here $\mathcal{L}_A$ denotes a particular local operator in the bulk integrated over $\gamma_A$: at leading order in Newton's constant $G$ we have $\LA=\frac{\mathrm{Area}(\gamma_A)}{4G}$, while at higher orders, both in $G$ but also in other couplings such as $\alpha'$, there are corrections to $\Ll_A$ involving various intrisic and extrinsic quantities integrated on $\gamma_A$ \cite{Wald:1993nt,Iyer:1994ys,Jacobson:1993vj,Solodukhin:2008dh,Hung:2011xb,Bhattacharyya:2013jma,Fursaev:2013fta,Dong:2013qoa,Camps:2013zua,Faulkner:2013ana,Miao:2014nxa}.  $S_{bulk}(\rho_{\EA})$ denotes the bulk von Neumann entropy in $\EA$.\footnote{Here I have been somewhat cavalier about how the surface $\gamma_A$ is to be chosen at higher orders in $G$.  This was worked out to first nontrivial order in \cite{Faulkner:2013ana}, and a conjecture for higher orders was given in \cite{Engelhardt:2014gca}.  In this paper I will focus on reproducing \eqref{flm} only to order $G^0$, except for some brief comments at the end. Most results should be generalizable in some form to higher orders using some version of the proposal of \cite{Engelhardt:2014gca}, see \cite{Dong:2016eik,donglew}.}  I will refer to the first term on the right hand side of \eqref{flm} as the ``area term'', and the second term as the ``bulk entropy term''.  I will also sometimes refer to $\LA$ as the ``area operator'', although this isn't strictly true.  One puzzling feature of \eqref{flm} is what precisely is meant by an ``appropriate'' state. Another is that the area term is linear in the state $\rho$, while the left hand side of \eqref{flm} is not: since the bulk entropy term is subleading in $G$ for states where geometric fluctuations are small, this has sometimes led to the suggestion that the RT formula violates the linearity of quantum mechanics \cite{Papadodimas:2015jra,Almheiri:2016blp}.  
\ei

In \cite{Almheiri:2014lwa} it was explained how the first two of these properties are naturally realized in quantum error correction: radial commutativity illustrates the fact that no particular boundary point is indispensible for a CFT representation of the bulk operator $\phi(x)$, and subregion duality illustrates the ability of the code to correct the operator $\phi(x)$ for the erasure of a region $\ol{A}$, provided that $x$ lies in $\EA$.  In \cite{Almheiri:2014lwa,Jafferis:2014lza,Jafferis:2015del} it was suggested that the RT formula might actually imply subregion duality in the entanglement wedge, in \cite{Pastawski:2015qua,Hayden:2016cfa} the RT formula and subregion duality were both confirmed in some tensor network models of holography, and in \cite{Dong:2016eik} the implication RT $\Rightarrow$ subregion duality was proven using techniques from quantum error correction, as well as the results of \cite{Jafferis:2015del}.  For all three properties, a key point is that they hold only on a \textit{code subspace} of states, which roughly speaking must be chosen to ensure that bulk effective field theory is a good approximation for the observables of interest throughout the subspace.  Restricting the validity of our three properties to this subspace is essential in explaining the paradoxical features of the correspondence mentioned above.    

So far the explanations of these properties and the relationships between them have been somewhat scattered. The goal of this paper is to tie them all together into a set of theorems which give a rather general picture of how quantum error correction realizes subregion duality and the RT formula.  I will first present a simple example that illustrates many of the results, and then gradually build up the machinery to deal with the most general case. 

As we proceed, it will become clear that von Neumann algebras are a language particularly suited for studying subregion duality and the RT formula.  The final results will thus be phrased in the language of the \textit{operator-algebra quantum error correction} of \cite{beny2007generalization,beny2007quantum}.  For the convenience of the reader, the discussion of von Neumann algebras will be completely self-contained, with proofs of the necessary theorems given in appendix \ref{vnapp}. The culmination of my analysis will be the following theorem: 
\begin{thm}\label{bigrthm}
Say that we have a (finite-dimensional) Hilbert space $\Hh=\HA\otimes \HAb$, a code subspace $\Hc\subseteq\Hh$, and a von Neumann algebra $M$ acting on $\Hc$.  Then the following three statements are equivalent:
\bi
\item There exists an operator $\LA\in Z_M\equiv M\cap M'$ such that, for any state $\wt{\rho}$ on $\Hc$, we have
\begin{align}\nonumber
S(\wt{\rho}_A)&=\Tr \left(\wt{\rho} \LA\right)+S(\wt{\rho},M)\\\nonumber
S(\wt{\rho}_{\Ab})&=\Tr \left(\wt{\rho} \LA\right)+S(\wt{\rho},M')
\end{align}

\item For any operators $\wt{O}\in M$, $\wt{O}'\in M'$, there exists operators $O_A$, $O_{\Ab}'$ on $\HA$, $\HAb$ respectively such that, for any state $|\wt{\psi}\ran\in \Hc$, we have
\begin{align}\nonumber
O_A|\wt{\psi}\ran&=\wt{O}|\wt{\psi}\ran\\\nonumber
O_A^\dagger|\wt{\psi}\ran&=\wt{O}^\dagger|\wt{\psi}\ran\\\nonumber
O_{\Ab}'|\wt{\psi}\ran&=\wt{O}'|\wt{\psi}\ran\\\nonumber
O_{\Ab}'^\dagger|\wt{\psi}\ran&=\wt{O}'^\dagger|\wt{\psi}\ran
\end{align}

\item For any two states $\wt{\rho}$, $\wt{\sigma}$ on $\Hc$, we have
\begin{align}\nonumber
S\left(\wt{\rho}_A|\wt{\sigma}_A\right)&=S(\wt{\rho}|\wt{\sigma},M)\\
S\left(\wt{\rho}_{\Ab}|\wt{\sigma}_{\Ab}\right)&=S(\wt{\rho}|\wt{\sigma},M')\nonumber
\end{align}
\ei
\end{thm}
Here $M'$ is the commutant of $M$ on $\Hc$, $S(\wt{\rho},M)$ denotes the algebraic entropy of the state $\wt{\rho}$ on $M$, and $S(\wt{\rho}|\wt{\sigma},M)$ denotes the relative entropy of $\wt{\rho}$ to $\wt{\sigma}$ on $M$. These concepts will be introduced in more detail as we go along.  In applying this theorem to AdS/CFT, we should think of $M$ as the algebra of bulk operators in $\EA$ and $M'$ as the algebra of bulk operators in $\mathcal{E}_{\Ab}$.  This theorem then shows the complete equivalence of the RT formula and subregion duality, and also shows their equivalence to the relative entropy relation of \cite{Jafferis:2015del}.\footnote{In \cite{Lashkari:2016idm}, it was shown that, in the special case of a spherical boundary region, the boundary relative entropy of a state to the vacuum is equivalent to the canonical energy in that region.  From the bulk point of view it is not obvious that this canonical energy is non-negative, so in \cite{Lashkari:2016idm} it was suggested that this is a constraint on low energy effective field theories.  The third condition of theorem \ref{bigrthm} suggests however that this constraint should be automatic for any state whose bulk relative entropy is non-negative: this should require only unitarity in the bulk effective field theory.  It would be very interesting to find a direct classical proof that canonical energy is positive starting from something like the dominant energy condition.}

\bfig
\includegraphics[height=5cm]{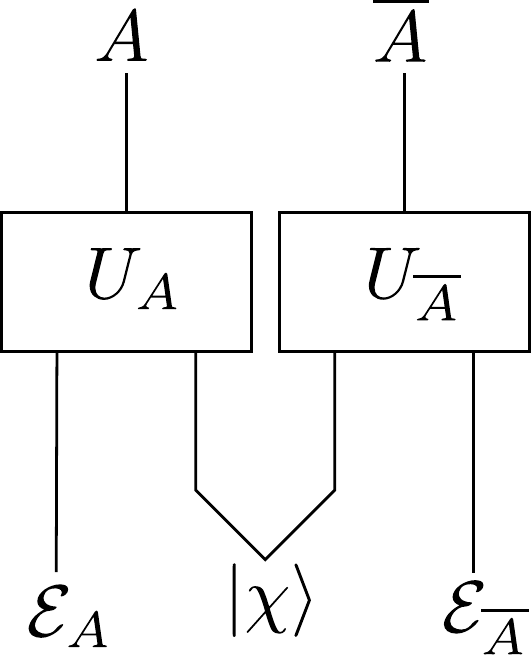}
\caption{A holographic encoding circuit.  $A$ is a CFT subregion, $\Ab$ is its complement, and $\EA$ and $\mathcal{E}_{\Ab}$ are the bulk degrees of freedom in their respective entanglement wedges.  We encode these bulk degrees of freedom into the CFT by acting with unitary transformations $U_A$ and $U_{\Ab}$ that mix $\EA$ and $\mathcal{E}_{\Ab}$ with complementary pieces of a fixed state $|\chi\ran$, which accounts for the remaining CFT degrees of freedom in $A$ and $\Ab$.  The entanglement in the state $|\chi\ran$ is the source of the area terms of the RT formulae for $S_A$ and $S_{\Ab}$, while the states that are fed into $\EA$ and $\mathcal{E}_{\Ab}$ give the bulk entropy terms.  We will see that nonvanishing entanglement in $|\chi\ran$, and thus a nonvanishing area term, is necessary for the robust functioning of the code.}\label{circuit}
\efig
On the way to proving this theorem, I will also introduce a ``completely boundary'' interpretation of the RT formula, which might be contrasted with the ``completely bulk'' explanation of \cite{Lewkowycz:2013nqa,Faulkner:2013ana}.  I sketch the basic idea in figure \ref{circuit} for the special case where the algebra $M$ is a factor, meaning that we take the code subspace to tensor-factorize into the degrees of freedom in $\EA$ and those in $\mathcal{E}_{\Ab}$.  This gives a circuit picture of how bulk information in the entanglement wedges is encoded into the CFT, with simple interpretations for both terms in the RT formula \eqref{flm}.  This picture is not quite satisfactory, in that the area operator it produces is a trivial operator proportional to the identity.  This is actually required by the properties of $\LA$ stated in theorem \ref{bigrthm}, since we saw there that $\LA$ must be in the center $Z_M$ of $M$, which is trivial if $M$ is a factor.  Fixing that problem is what leads us to consider general algebras.  Up to this subtlety, we will see that the setup of figure \ref{circuit} is not only sufficient for the RT formula and subregion duality to work, it is also necessary.  

The bulk of this paper is spent establishing theorem \ref{bigrthm} and the algebraic generalization of figure \ref{circuit}, but in a final discussion section we will see what these results imply for AdS/CFT. The basic points are: 
\bi
\item The observation that $\LA$ must be in the center of $M$ is consistent with the fact that the area operator is part of the ``edge modes''/``soft hair'' of \cite{Donnelly:2015hxa,Donnelly:2014fua,Hawking:2016msc}, and the nontriviality of this center is closely related to bulk gauge symmetry.  In \cite{Harlow:2015lma} these degrees of freedom were given a short-distance interpretation, which fits naturally into the quantum error correction picture I discuss here.   
\item Figure \ref{circuit} suggests a boundary interpretation of the ``bit threads'' that were recently used to give an alternative presentation of the RT formula \cite{Freedman:2016zud}.  This presentation is subtle in the multipartite case, but I give it a preliminary interpretation as well.
\item Figure \ref{circuit} also ensures that including the bulk entropy term in the RT formula removes any problems with linearity.  We will see that its algebraic version reproduces the nonlinear ``entropy of mixing'' studied in \cite{Papadodimas:2015jra,Almheiri:2016blp}, and that it also gives a new perspective on the ``homology constraint'' often included in the definition of the HRT surface $\gamma_A$ \cite{Headrick:2007km,Haehl:2014zoa}. In \cite{Almheiri:2016blp} it was recently argued that the homology constraint is sometimes inconsistent with the linearity of quantum mechanics, but we'll see that figure \ref{circuit} requires that we do not include this constraint in such situations: the bulk entropy term in \eqref{flm} is able to make up the difference without violating linearity.  In particular we will see that there is no obstruction to the the RT formula holding in superpositions of states with different classical geometries.
\item In general there is a close connection between changing the size of the code subspace and renormalization group flow in the bulk: including more UV degrees of freedom in the bulk has long been expected to shift entropy from the area term to the bulk entropy term, see eg \cite{Solodukhin:2011gn} for a review, and quantum error correction formalizes this operation as the inclusion of more states in the code subspace.  In figure \ref{circuit}, doing this moves degrees of freedom from $|\chi\ran$ to $\EA$ and $\EAb$, which indeed decreases the area term.
\ei
The structure of this paper is that I first present a simple example, then prove the main theorems, and then explain these points in more detail in a final discussion.  Readers who are willing to accept theorem \ref{bigrthm} and figure \ref{circuit} without proof may wish to proceed directly to this discussion, which should already be mostly comprehensible, although studying the example in section \ref{exsec} first wouldn't hurt.

\subsection{Notation}  My notation will at times be a bit heavy, so I will lay out a few rules here.  I will label physical systems by roman letters, eg $A$, $a$, $R$, etc, and their associated Hilbert spaces as $\HA$,$\Ha$, $\Hh_R$, etc.  Upper case letters will refer to subsystems of the full physical Hilbert space $\Hh$, while lowercase letters will refer to subsystems (or subsystems of subspaces) of the code subspace $\Hc$.  I will write $|A|$ for the dimensionality of $\HA$, $|R|$ for the dimensionality of $\Hh_R$, etc.  I will often indicate with subscripts which Hilbert space a state lives in or an operator acts on; for example $|\psi\ran_A$ is an element of $\HA$, and $O_A$ is a linear operator on $\HA$.  I will sometimes abuse notation by neglecting to write the identity factors which are technically needed to lift the action of an operator on a subfactor of a Hilbert space to an operator on the whole Hilbert space.  For example in stating theorem \ref{bigrthm} I did not distinguish between $O_A\otimes I_{\Ab}$ and $O_A$.  In any particular equation it should be straightforward to supply the identity factors as needed to ensure that all operators act on the correct spaces.  I will use the ``tilde'' symbol on operators which are naturally defined to act within the code subspace $\Hc$, although I have had to make arbitrary choices in a few places where it isn't so clear what is ``natural''.  Finally, whenever I say an operator ``acts within a subspace'', I always mean that both the operator and its hermitian conjugate act within the subspace.

\section{An example}\label{exsec}
I'll begin with a simple example that illustrates many of the ideas of this paper: the three-qutrit code of \cite{Cleve:1999qg}.  This code was first used as a model of holography in \cite{Almheiri:2014lwa}, and despite its simplicity, it captures many features of quantum gravity. Indeed it has analogues of effective field theory, black holes, radial commutativity, subregion duality, and the RT formula!

The basic idea of quantum error correction is to protect a quantum state by encoding it into a \textit{code subspace} of a larger Hilbert space. The three-qutrit code is an encoding of a single ``logical'' qutrit into the Hilbert space of three ``physical'' qutrits, with the code subspace $\Hh_{code}$ carrying the logical qutrit spanned by the basis
\begin{align}\nonumber
|\wt{0}\ran&\equiv \frac{1}{\sqrt{3}}\left(|000\ran+|111\ran+|222\ran\right)\\\nonumber
|\wt{1}\ran&\equiv \frac{1}{\sqrt{3}}\left(|012\ran+|120\ran+|201\ran\right)\\\nonumber
|\wt{2}\ran&\equiv \frac{1}{\sqrt{3}}\left(|021\ran+|102\ran+|210\ran\right).\label{3qutrit}
\end{align}
This subspace has the property that there exists a unitary $U_{12}$, supported only on the first two qutrits, which obeys
\be\label{3decode}
U_{12}^\dagger|\wt{i}\ran=|i\ran_1 |\chi\ran_{23},
\ee
with 
\be
|\chi\ran\equiv \frac{1}{\sqrt{3}}\left(|00\ran+|11\ran+|22\ran\right).
\ee
This unitary is easy to find, and is described explicitly in \cite{Almheiri:2014lwa}.  Its existence enables this code to protect the state of the logical qutrit against the erasure of the third physical qutrit.  Indeed say that I wish to send you the single-qutrit state
\be
|\psi\ran=\sum_{i=0}^2 C_i |i\ran.
\ee
If I simply send it to you using a single qutrit, it could easily be corrupted.  But if I instead send you the three-qutrit state
\be
|\wt{\psi}\ran=\sum_{i=0}^2 C_i |\wt{i}\ran,
\ee
then even if the third qutrit is lost, you can use your handy quantum computer to apply $U_{12}^\dagger$ to the two qutrits you do receive, which allows you to recover the state on the first qutrit:
\be
U_{12}^\dagger|\wt{\psi}\ran=|\psi\ran_1|\chi\ran_{23}.
\ee
Moreover the symmetry between the qutrits in the definition of $\Hc$ ensures that unitaries $U_{13}$ and $U_{23}$ will also exist, which means that the state $|\psi\ran$ can be recovered on any two of the qutrits.  

We can also phrase this correctability of single-qutrit erasures in terms of operators. Say that $O$ is a linear operator on the single-qutrit Hilbert space.  We can easily find a three-qutrit operator $\wt{O}$ that acts within $\Hc$ with the same matrix elements as $O$, but if we extend this operator arbitrarily on the orthogonal complement $\Hc^\perp$, then it will in general define an operator with support on all three physical qutrits. Using $U_{12}$ however, we can define an operator
\be
O_{12}\equiv U_{12} O_1 U_{12}^\dagger
\ee   
that acts within $\Hc$ in the same way as $\wt{O}$ but has support only on the first two qutrits.  Again by symmetry we can also define an $O_{13}$ and $O_{23}$, so any logical operator on the code subspace can be represented as an operator with trivial support on any one of the physical qutrits.  

Now say that we have an arbitrary mixed state $\wt{\rho}$ on $\Hc$, which is the encoding of a ``logical'' mixed state $\rho$.  From eq. \eqref{3decode}, we see that 
\be
\wt{\rho}=U_{12}\left(\rho_1 \otimes |\chi\ran\lan\chi|_{23}\right) U_{12}^\dagger,
\ee
so defining $\wt{\rho}_3\equiv \Tr_{12}\wt{\rho}$ and $\wt{\rho}_{12}\equiv \Tr_3 \wt{\rho}$, we have the von Neumann entropies
\begin{align}\nonumber
S(\wt{\rho}_3)&=\log 3\\
S(\wt{\rho}_{12})&=\log 3+S(\wt{\rho}).\label{3ent}
\end{align}
Once again, the symmetry ensures that analogous results hold for the entropies on other subsets of the qutrits.

We can interpret this code as a model of AdS/CFT.  The three physical qutrits are analogous to the local CFT degrees of freedom, and the code subspace $\Hc$ is analogous to the subspace where only effective field theory degrees of freedom are excited in the bulk.  This ``bulk effective field theory'' has only one spatial point, at which we have a single qutrit.  We can illustrate this using the right diagram of figure \ref{subregions}, where now $A$, $B$, and $C$ denote the three physical qutrits and $x$ denotes our bulk point.  The orthogonal complement $\Hc^\perp$  corresponds to the microstates of a black hole which has swallowed our point.  Let's now see how this realizes the properties of AdS/CFT discussed in the introduction:
\bi
\item \textbf{Radial Commutativity:}  We'd like to show that any ``bulk local operator'', meaning any operator $\wt{O}$ that acts within $\Hc$, commutes with all ``local operators at the boundary'', meaning it commutes with any operator that acts on only one physical qutrit. But $O_{12}$, $O_{13}$, and $O_{23}$ each manifestly commute with boundary local operators on the third, second, or first qutrits respectively, and since they all act identically to $\wt{O}$ within the code subspace, it must be that within the code subspace $\wt{O}$ commutes with all boundary local operators. More precisely, if $X$ is an operator on a single physical qutrit, and $|\wt{\psi}\ran$,$|\wt{\phi}\ran \in \Hc$, then $\lan \wt{\psi}|[\wt{O},X]|\wt{\phi}\ran=0$.
\item \textbf{Subregion Duality:} According to figure \ref{subregions}, we should think of $x$ as being in the entanglement wedge of any two of the boundary qutrits.  And indeed we see that any operator $\wt{O}$ can be represented on any two of the qutrits using $O_{12}$, $O_{13}$, or $O_{23}$.
\item \textbf{Ryu-Takayanagi Formula:} We have already computed the entropies \eqref{3ent}. If we define an ``area operator'' $\Ll_{12}=\Ll_{3}\equiv \log 3$, then apparently the RT formula \eqref{flm} holds for any state $\wt{\rho}$ on the code subspace.  This ``area term'' reflects the nontrivial entanglement in the state $|\chi\ran$, while the ``bulk entropy term'' takes into account the possibility of the encoded qutrit being in a mixed state.  The area term is essential for the functioning of the code, since if $|\chi\ran$ were a product state, from \eqref{3decode} we see that the third qutrit would be extemporaneous, and there would be no way for both $U_{23}$ and $U_{13}$ to exist (one of them could exist if the first or second qutrit could access the state by itself).  
\ei 
The three-qutrit code is thus able to capture a considerable amount of the physics of AdS/CFT.  

In fact this is more than an analogy, AdS/CFT itself can be recast in similar language.  To do this, we need to develop a general theory about when the analogue of $U_{12}$ exists and what its consequences are.  In the next three sections we will extend the basic features of the three-qutrit code via a set of theorems of increasing generality: purists may wish to skip directly to section \ref{opqecsec}, since the results obtained there contain the results of sections \ref{csec}, \ref{subsec} as special cases.

\section{Conventional quantum erasure correction}\label{csec}
The conventional version of quantum error correction is based on generalizing eq. \eqref{3decode}: we ask for the ability to recover an arbitrary state in the code subspace.  In general there are a variety of errors which can be considered, but in this paper I will study only \textit{erasures}, which are defined as losing access to a known subset of the physical degrees of freedom.  The three qutrit code was able to correct single-qutrit erasures.  There is a standard set of conditions which characterize whether or not a code can correct for any particular erasure \cite{Schumacher:1996dy,Grassl:1996eh}.  These can be gathered together into a theorem, which I'll now describe and prove.
\subsection{A theorem}
\begin{thm}\label{cthm}
Say that $\Hh$ is a finite-dimensional Hilbert space, with a tensor product structure $\Hh=\HA\otimes\HAb$, and say that $\Hc$ is a subspace in $\Hh$.  Moreover say that $|\wt{i}\ran$ is some orthonormal basis for $\Hc$, and that $|\phi\ran\equiv \frac{1}{\sqrt{|R|}} \sum_i |i\ran_R|\wt{i}\ran_{A\ol{A}}$, where $|i\ran_R$ denotes an orthonormal basis for an auxiliary system $R$ whose dimensionality $|R|$ is equivalent to that of $\Hc$.  Then the following statements are equivalent:
\bi
\item[(1)] $|R|\leq |A|$, and if we decompose $\HA=(\Hh_{A_1}\otimes\Hh_{A_2})\oplus \Hh_{A_3}$, with $|A_1|=|R|$ and $|A_3|<|R|$, then there exists a unitary transformation $U_A$ on $\HA$ and a state $|\chi\ran_{A_2\ol{A}}\in \Hh_{A_2\ol{A}}$ such that
\be
|\wt{i}\ran=U_A\left(|i\ran_{A_1}\otimes |\chi\ran_{A_2\ol{A}}\right),
\ee
where $|i\ran_{A_1}$ is an orthonormal basis for $\Hh_{A_1}$.

\item[(2)] For any operator $\wt{O}$ acting within $\Hc$, there exists an operator $O_A$ on $\HA$ such that, for any state $|\wt{\psi}\ran\in \Hc$, we have
\begin{align}\nonumber
O_A|\wt{\psi}\ran&=\wt{O}|\wt{\psi}\ran\\
O_A^\dagger|\wt{\psi}\ran&=\wt{O}^\dagger|\wt{\psi}\ran.
\end{align}

\item[(3)] For any operator $X_{\ol{A}}$ on $\HAb$, we have
\be
P_{code}X_{\ol{A}}P_{code}\propto P_{code}.
\ee
Here $P_{code}$ denotes the projection onto $\Hc$. 
\item[(4)] In the state $|\phi\ran$, we have
\be
\rho_{R\Ab}(\phi)=\rho_R(\phi)\otimes \rho_{\Ab}(\phi).
\ee

\ei
\end{thm}

Condition (1) is the statement that we can recover the full state of the code subspace on $A_1$ by applying $U_{A}^\dagger$, while condition (2) says that any logical operator on the code subspace can be represented by an operator on $A$.  Condition (3) says that measuring any operator on the erased subsystem cannot disturb the encoded information, while condition (4) says that there is no correlation between the operators on the reference system $R$ and operators on the erased subsystem $\ol{A}$.  Each of these conditions is quite plausibly necessary for the correctability of the erasure of $\Ab$.  Their equivalence can be proven as follows:
\begin{proof}

$(1)\Rightarrow (2)$: Defining $O_A\equiv U_A O_{A_1} U_A^\dagger$, the claimed properties are immediate.  Here $O_{A_1}$ is an operator on $A_1$ that acts with the same matrix elements as $\wt{O}$ does on the code subspace.

$(2)\Rightarrow(3)$:  Say that there were an $X_{\Ab}$ such that $P_{code}X_{\Ab}P_{code}$ was not proportional to $P_{code}$.  By Schur's lemma, there then must be an operator $\wt{O}$ on $\Hc$ and a state $|\wt{\psi}\ran\in \Hc$ such that $\lan\wt{\psi}|[P_{code}X_{\Ab}P_{code},\wt{O}]|\wt{\psi}\ran=\lan\wt{\psi}|[X_{\Ab},\wt{O}]|\wt{\psi}\ran\neq 0$.  But clearly this $\wt{O}$ cannot have a representation $O_A$ on $\HA$, since this would automatically commute with $X_{\Ab}$.  Therefore no such $X_{\Ab}$ can exist.  

$(3)\Rightarrow(4)$: Consider an arbitary operator $O_R$ on $\mathcal{H}_R$ and an arbitrary operator $X_{\Ab}$ on $\HAb$.  By (3), we must have $\Pc X_{\Ab} \Pc=\lan\phi|X_{\Ab}|\phi\ran \Pc$.  But this implies that
\begin{align}\nonumber
\lan\phi|X_{\Ab}O_R|\phi\ran&=\lan\phi|O_R \Pc X_{\Ab}\Pc |\phi\ran\\
&=\lan\phi|X_{\Ab}|\phi\ran\lan\phi|O_R|\phi\ran.
\end{align}
If $|\phi\ran$ has no nonvanishing connected correlation function for any operators $O_R$, $X_{\Ab}$, then $\rho_{R\Ab}(\phi)=\rho_{R}(\phi)\otimes \rho_{\Ab}(\phi)$. 

$(4) \Rightarrow (1)$: First note that $|\phi\ran$ is a purification of $\rho_{R\Ab}=\rho_{R}\otimes \rho_{\Ab}$. Such a purification is only possible if $|R|$ times the rank of $\rho_{\Ab}$ is less than or equal to $|A|$,\footnote{This statement follows immediately from the Schmidt decomposition, as do several more of the implications in this proof.  The Schmidt decomposition says that for any bipartite pure state $|\psi\ran_{A\Ab}$, there are sets of orthonormal states $|n\ran_A$, $|n\ran_{\Ab}$, such that $|\psi\ran_{A\Ab}=\sum_n \sqrt{p_n} |n\ran_A|n\ran_{\Ab}$, with $p_n\geq 0$.  These orthonormal states are eigenstates of the density matrices on $\rho_A$ and $\rho_{\Ab}$, which have equal nonzero eigenvalues given by the positive $p_n$'s.}  so indeed $|R|\leq |A|$.  Long division of $|A|$ by $|R|$ gives $|A_2|$ and $|A_3|$ such that we can decompose $\HA$ as in (1).  Since $|A_3|\leq|R|-1$, we see that the rank of $\rho_{\Ab}$ can be at most $|A_2|$.  Therefore another purification of $\rho_{R\Ab}$ is given by
\be
|\phi'\ran=\left(\frac{1}{\sqrt{|R|}}\sum_i|i\ran_R|i\ran_{A_1}\right)\otimes |\chi\ran_{A_2\ol{A}},
\ee
where $|\chi\ran_{A_2\Ab}$ is an arbitrary purification of $\rho_{\Ab}$ on $A_2$.  But any two purifications of the same density matrix onto the same additional system differ only by a unitary transformation on that system, so we must have $|\phi\ran=U_A|\phi'\ran$ for some $U_A$ on $A$.  This then implies (1).

\end{proof} 
This theorem gives several useful conditions to diagnose whether or not the erasure of $\Ab$ is correctable in the conventional sense of complete state recovery.  One thing it does not fully characterize however is the full set of erasures that can be corrected by a given code subspace; we just need to apply the theorem separately for each erasure and hope for the best.  For example the three qutrit code could correct for any single-qutrit erasure, but that isn't obvious from a particular decomposition into $A$ and $\Ab$.  We saw in the previous section however that this robustness of the code was a consequence of the nonzero entanglement in the state $|\chi\ran_{23}$.  The same is true here: if $|\chi\ran_{A_2 \Ab}$ is a product state, then we can dispense with $\Ab$ entirely.  It is only when $|\chi\ran$ is entangled that we can have a situation where a subsystem of $A$ together with $\Ab$ might be able to access encoded information which that subsystem by itself cannot.  

\subsection{A Ryu-Takayanagi formula}
We can see immediately from condition (1) of theorem \ref{cthm} that, if the erasure of $\Ab$ is correctable, then for any mixed state $\wt{\rho}$ on the code subspace we have
\begin{align}
\wt{\rho}&=U_A\Big(\rho_{A_1}\otimes |\chi\ran\lan\chi|_{A_2\Ab}\Big)U_A^\dagger\\\label{convra}
\wt{\rho}_A&\equiv \Tr_{\Ab}\wt{\rho}=U_A\Big(\rho_{A_1}\otimes \Tr_{\Ab}\left(|\chi\ran\lan \chi|\right)\Big)U_A^\dagger\\
\wt{\rho}_{\Ab}&\equiv \Tr_{A}\wt{\rho}=\Tr_{A_2}\left(|\chi\ran\lan \chi|\right)
\end{align}
Here $\rho_{A_1}$ is an operator on $\mathcal{H}_{A_1}$ with the same matrix elements as $\wt{\rho}$ on $\Hc$.  Defining $\chi_{A_2}\equiv \Tr_{\Ab}|\chi\ran\lan\chi|$ and $\chi_{\Ab}\equiv \Tr_{A_2}|\chi\ran\lan\chi|$, we see that
\begin{align}\label{convrt1}
S\left(\wt{\rho}_A\right)&=S\left(\chi_{A_2}\right)+S(\wt{\rho})\\\label{convrt2}
S\left(\wt{\rho}_{\Ab}\right)&=S\left(\chi_{A_2}\right).
\end{align}
If we define an ``area operator''
\be
\Ll_A\equiv S\left(\chi_{A_2}\right) I_{code},
\ee 
then eqs. \eqref{convrt1}, \eqref{convrt2} are reminiscent of the RT formula eq. \eqref{flm}.  The analogy is not perfect, as we will discuss momentarily, but notice that the area term arises from the nontrivial entanglement in $|\chi\ran$, which we just saw is necessary for the robustness of the code.

Condition (1) also has interesting consequences for the modular Hamiltonians $\wt{K}^\rho\equiv -\log \wt{\rho}$, $\wt{K}^\rho_A\equiv -\log \wt{\rho}_A$, and $\wt{K}^\rho_{\Ab}\equiv -\log \wt{\rho}_{\Ab}$.  Applying the identity $e^X\otimes e^Y=e^{X\otimes I+I\otimes Y}$ to eq. \eqref{convra}, we see that 
\be
\wt{K}^\rho_A=U_A\left(K_{A_1}^\rho\otimes I_{A_2}-I_{A_1}\otimes \log \chi_{A_2}\right)U_A^\dagger.
\ee
Using this together with the code subspace projection
\be
\Pc \equiv \sum_i|\wt{i}\ran\lan\wt{i}|,
\ee
we see that
\be\label{cjlms1}
\Pc\wt{K}^\rho_A \Pc=\wt{K}^\rho+\Ll_A.
\ee
Similarly we can show that
\be\label{cjlms2}
\Pc\wt{K}^\rho_{\Ab}\Pc=\Ll_A.
\ee
These expressions are analagous to the main result of \cite{Jafferis:2015del}, which said that the boundary modular Hamiltonian of a subregion $A$ is equal to the bulk modular Hamiltonian in $\EA$ plus the area operator $\Ll_A$.  This was originally derived directly from the RT formula \cite{Jafferis:2015del,Dong:2016eik}, but we see here that it is also a direct consequence of correctability. 
 
\subsection{Some problems}\label{probsec}
In the previous section we found that ``RT-like'' formulae \eqref{convrt1}, \eqref{convrt2} hold for any conventional quantum erasure-correcting code. But the bulk entropy term did not appear symmetrically in these results: all of the ``bulk entropy'' $S(\wt{\rho})$ appeared in $S(\wt{\rho}_A)$, while none appeared in $S(\wt{\rho}_{\Ab})$.  This is a consequence of insisting that we can recover the entire state on $A$: this was ok when the bulk only had one point, as in the example of section \ref{exsec}, but it will obviously not be true in more realistic examples of holography where the entanglement wedge of $\Ab$ is nontrivial. 

A related problem with this formalism was identified in \cite{Almheiri:2014lwa}: consider the situation of the left diagram in figure \ref{subregions}.  We might want to view the operator $\phi(x)$ as an operator on the code subspace, which can be reconstructed on $A$ as in condition (2). But in the ground state $|0\ran$ this operator has nonzero correlation with the operator $\phi(y)$, which we should be able to reconstruct on $\Ab$.  This contradicts condition (4), which would imply that there can be no correlation between operators on the code subspace and operators $X_{\Ab}$ on the erased region $\Ab$.  

Both of these issues tell us that conventional quantum erasure correction, as characterized by theorem \ref{cthm}, needs to be generalized to simultaneously allow some information to be recovered on $A$ and other information to be recovered on $\Ab$.  We can realize this by a generalization of quantum erasure correction which I will now describe.

\section{Subsystem quantum erasure correction}\label{subsec}
A generalization of quantum error correction that allows for the physical degrees of freedom in $A$ to access only partial information about the encoded state has existed in the coding literature for some time \cite{kribs2005unified,kribs2005operator,nielsen2007algebraic}.  It was originally called ``operator quantum error correction'', but since this term is unfortunately similar to the more general ``operator-algebra quantum error correction'' I will present in the next section, I will instead refer to the framework of \cite{kribs2005unified,kribs2005operator,nielsen2007algebraic} as \textit{subsystem quantum error correction}.  The basic idea is to consider a code subspace which factorizes as $\Hc=\Ha\otimes \Hab$, and then only ask for recovery of the state of $\Ha$.  For erasure errors, the results of \cite{kribs2005unified,kribs2005operator,nielsen2007algebraic} can be combined into a theorem analogous to theorem \ref{cthm} for conventional codes.     

\subsection{A theorem}
\begin{thm}\label{sthm}
Say that $\Hh$ is a finite-dimensional Hilbert space, with a tensor product structure $\Hh=\HA\otimes\HAb$, and say that $\Hc$ is a subspace of $\Hh$ which factorizes as $\Hc=\Ha\otimes\Hab$.  Moreover say that $|\wt{i}\ran$ is some orthonormal basis for $\Ha$, that $|\wt{j}\ran$ is some orthonormal basis for $\Hab$, and that $|\phi\ran\equiv \frac{1}{\sqrt{|R||\ol{R}|}} \sum_{i,j} |i\ran_R|j\ran_{\ol{R}}|\wt{ij}\ran_{A\ol{A}}$, where $R$ and $\ol{R}$ are auxiliary systems whose dimensionalities are equal to those of $a$ and $\ab$ respectively.  Then the following statements are equivalent:
\bi
\item[(1)] $|a|\leq |A|$, and if we decompose $\HA=(\Hh_{A_1}\otimes\Hh_{A_2})\oplus \Hh_{A_3}$, with $|A_1|=|a|$ and $|A_3|<|a|$, there exists a unitary transformation $U_A$ on $\HA$ and a set of orthonormal states $|\chi_j\ran_{A_2\ol{A}}\in \Hh_{A_2\ol{A}}$ such that
\be
|\wt{ij}\ran=U_A\left(|i\ran_{A_1}\otimes |\chi_j\ran_{A_2\ol{A}}\right),
\ee
where $|i\ran_{A_1}$ is an orthonormal basis for $\Hh_{A_1}$.

\item[(2)] For any operator $\wt{O}_a$ acting within $\Ha$, there exists an operator $O_A$ on $\HA$ such that, for any state $|\wt{\psi}\ran\in \Hc$, we have
\begin{align}\nonumber
O_A|\wt{\psi}\ran&=\wt{O}_a|\wt{\psi}\ran\\
O_A^\dagger|\wt{\psi}\ran&=\wt{O}^\dagger_a|\wt{\psi}\ran.
\end{align}

\item[(3)] For any operator $X_{\ol{A}}$ on $\HAb$, we have
\be
P_{code}X_{\ol{A}}P_{code}=(I_a\otimes X_{\ab})\Pc
\ee
with $X_{\ab}$ an operator on $\Hab$.  Here $P_{code}$ again denotes the projection onto $\Hc$. 
\item[(4)] In the state $|\phi\ran$, we have
\be
\rho_{R\ol{R}\,\Ab}(\phi)=\rho_R(\phi)\otimes \rho_{\ol{R}\,\Ab}(\phi).
\ee

\ei
\end{thm}
This theorem gives a broad characterization of when a code can recover the state of a logical subsystem $a$ from the erasure of a physical subsystem $\Ab$.  The proof is quite similar to the proof of theorem \ref{cthm}, one just needs to keep track of $\Ha$, so I won't give the details here (anyways it is a special case of the analogous theorem in the next section).

In applying this theorem to AdS/CFT, we are mostly interested in the special case where, in addition to being able to recover an arbitrary $\wt{O}_a$ on $A$, we can also recover an arbitrary $\wt{O}_{\ab}$ on $\Ab$ (see the left diagram of figure \ref{subregions}).  I'll call this a \textit{subsystem code with complementary recovery}.\footnote{This criterion seems related to the ``quantum mutual independence'' of \cite{horodecki2009quantum}, I thank Jonathan Oppenheim for bringing this to my attention.  The explicit examples given here suggest that quantum mutual independence is more common than was suggested in \cite{horodecki2009quantum}, it would be interesting to understand this better.}  This restriction implies that condition (1) of theorem \ref{sthm} should apply also for the barred factors:
\be\label{ijeq}
|\wt{ij}\ran=U_A\left(|i\ran_{A_1}|\chi_j\ran_{A_2 \Ab}\right)=U_{\Ab}\left(|j\ran_{\Ab_1}|\ol{\chi}_i\ran_{\Ab_2 A}\right).
\ee
Here we have decomposed $\HAb=\left(\Hh_{\Ab_1}\otimes\Hh_{\Ab_2}\right)\oplus \Hh_{\Ab_3}$, with $|\Ab_1|=|\ol{R}|=|\ab|$ and $|\Ab_3|<|\ab|$, $|j\ran_{\Ab_1}$ is an orthonormal basis for $\Hh_{\Ab_1}$, and the states $|\ol{\chi}_i\ran_{\Ab_2 A}$ are orthonormal.  Acting on \eqref{ijeq} with $U_A^\dagger U_{\Ab}^\dagger$, we see that we must have states $|\chi\ran_{A_2\Ab_2}$, $|\ol{\chi}\ran_{A_2\Ab_2}$ such that
\begin{align}\nonumber
U_{\Ab}^\dagger|\chi_j\ran_{A_2 \Ab}&=|j\ran_{\Ab_1}|\chi\ran_{A_2\Ab_2}\\
U_A^\dagger|\ol{\chi}_i\ran_{\Ab_2 A}&=|i\ran_{A_1}|\ol{\chi}\ran_{A_2\Ab_2},
\end{align}
which together with \eqref{ijeq} imply that actually $|\chi\ran_{A_2\Ab_2}=|\ol{\chi}\ran_{A_2\Ab_2}$.  Thus we must have
\be\label{sstate}
|\wt{ij}\ran=U_A U_{\Ab}\left(|i\ran_{A_1}|j\ran_{\Ab_1}|\chi\ran_{A_2\Ab_2}\right).
\ee
This is precisely the situation illustrated by figure \ref{circuit} in the introduction, but now we see that it is really necessary for subregion duality to work with a factorized code subspace.

It is worth mentioning that the tensor-network models of holography introduced in \cite{Pastawski:2015qua,Hayden:2016cfa} provide explicit examples of subsystem codes with complementary recovery, so all results of this section apply to them. 

\subsection{A Ryu-Takayanagi formula}
Using eq. \eqref{sstate}, we can again study the entropy of any state $\wt{\rho}$ on $\Hc$ for a subsystem code with complementary recovery.  Defining $\chi_{A_2}\equiv \Tr_{\Ab_2}|\chi\ran\lan\chi|$ and $\chi_{\Ab_2}\equiv \Tr_{A_2}|\chi\ran\lan\chi|$, we now have  
\begin{align}
\wt{\rho}&=U_AU_{\Ab}\Big(\rho_{A_1\Ab_1}\otimes |\chi\ran\lan\chi|\Big)U_A^\dagger U_{\Ab}^\dagger\\\label{sra}
\wt{\rho}_A&\equiv \Tr_{\Ab} \wt{\rho}=U_A\Big(\rho_{A_1}\otimes \chi_{A_2}\Big)U_A^\dagger\\
\wt{\rho}_{\Ab}&\equiv \Tr_{A} \wt{\rho}=U_{\Ab}\Big(\rho_{\Ab_1}\otimes \chi_{\Ab_2}\Big)U_{\Ab}^\dagger.
\end{align}
Here $\rho_{A_1\Ab_1}$ acts within $\Hc$ with the same matrix elements as $\wt{\rho}$, and $\rho_{A_1}$ and $\rho_{\Ab_1}$ have the same matrix elements as $\wt{\rho}_a$ and $\wt{\rho}_{\ab}$ respectively.   Defining ``area operators'' 
\begin{align}
\Ll_A&\equiv  S(\chi_{A_2})I_a\\
\Ll_{\Ab}&\equiv S(\chi_{A_2})I_{\ab},
\end{align}
we then see that 
\begin{align}\label{sflm1}
S\left(\wt{\rho}_A\right)&=\Tr\left(\wt{\rho}_a \Ll_A\right)+S\left(\wt{\rho}_a\right)\\\label{sflm2}
S\left(\wt{\rho}_{\Ab}\right)&=\Tr\left(\wt{\rho}_{\ab} \Ll_{\Ab}\right)+S\left(\wt{\rho}_{\ab}\right).
\end{align}
Thus the RT formula \eqref{flm} holds exactly for any subsystem code with complementary recovery!

We can also extend the relationships \eqref{cjlms1}, \eqref{cjlms2} between ``bulk'' and ``boundary'' modular Hamiltonians to subsystem codes with complementary recovery.  Defining $\wt{K}^\rho_A\equiv -\log \wt{\rho}_A$, $\wt{K}^\rho_{\Ab}\equiv -\log \wt{\rho}_{\Ab}$, $\wt{K}^\rho_a\equiv -\log \wt{\rho}_a$, and $\wt{K}^\rho_{\ab}\equiv -\log \wt{\rho}_{\ab}$, we again can straightforwardly confirm that
\begin{align}
\wt{K}^\rho_A&=U_A\left(K^\rho_{A_1}\otimes I_{A_2}-I_{A_1}\otimes \log \chi_{A_2}\right)U_A^\dagger\\
\wt{K}^\rho_{\Ab}&=U_{\Ab}\left(K^\rho_{\Ab_1}\otimes I_{\Ab_2}-I_{\Ab_1}\otimes \log \chi_{\Ab_2}\right)U_{\Ab}^\dagger,
\end{align}
and thus that\footnote{In these equations my neglect of identity factors may be confusing, including them we have
\begin{align}\nonumber
\Pc \left(\wt{K}^\rho_A\otimes I_{\Ab}\right) \Pc&=\left(\wt{K}^\rho_a+\Ll_A\right)\otimes I_{\ab}\\\nonumber
\Pc \left(I_A\otimes \wt{K}^\rho_{\Ab}\right) \Pc&=I_a\otimes \left(\wt{K}^\rho_{\ab}+\Ll_{\Ab}\right).
\end{align}}
\begin{align}\label{sjlms}
\Pc \wt{K}^\rho_A \Pc&=\wt{K}^\rho_a+\Ll_A\\\label{sjlms1}
\Pc \wt{K}^\rho_{\Ab} \Pc&=\wt{K}^\rho_{\ab}+\Ll_{\Ab}.
\end{align}
This then implies a nice result about the ``bulk'' and ``boundary'' relative entropies of two states $\wt{\rho}$, $\wt{\sigma}$:
\begin{align}\nonumber
S(\wt{\rho}_A|\wt{\sigma}_A)&\equiv -S(\wt{\rho}_A)+\Tr(\wt{\rho}_A \wt{K}^{\sigma}_A)\\\nonumber
&=-S(\wt{\rho}_a)+\Tr(\wt{\rho}_a \wt{K}^\sigma_a)\\
&=S(\wt{\rho}_a|\wt{\sigma}_a),\label{sjlms2}
\end{align}
and similarly 
\be\label{sjlms3}
S(\wt{\rho}_{\Ab}|\wt{\sigma}_{\Ab})=S(\wt{\rho}_a|\wt{\sigma}_{\ab}).
\ee  
In AdS/CFT, \eqref{sjlms}, \eqref{sjlms1}, \eqref{sjlms2}, and \eqref{sjlms3} are precisely the main results of \cite{Jafferis:2015del}; we now see they are general consequences of subsystem coding with complementary recovery.

\subsection{Holographic interpretation}
By now it should be clear that subsystem codes with complementary recovery resolve both of the problems mentioned in sec. \ref{probsec}. The new RT formulae, \eqref{sflm1}, \eqref{sflm2}, are symmetric between $A$ and $\Ab$, and allow for bulk information in both of their entanglement wedges.  Moreover in states with entanglement between $\Ha$ and $\Hab$, there can be nontrivial bulk correlation without violating any of the conditions of theorem \ref{sthm}.

In fact these RT formulae give a converse to the ``reconstruction theorem'' proven in \cite{Dong:2016eik}: there it was argued that if \eqref{sflm1}, \eqref{sflm2} hold for some operators $\Ll_A$ and $\Ll_{\Ab}$ in all state $\wt{\rho}$ on a factorized code subspace $\Hc=\Ha\otimes \Hab$, then condition (2) of theorem \ref{sthm} also holds.  But now we have learned something new: we also must have
\be
\Ll_A\otimes I_{\ab}=I_{a}\otimes \Ll_{\Ab}=S(\chi_{A_2}) I_{code}.
\ee  
This is rather unsettling: in AdS/CFT, the area operator is certainly not trivial!  We can check this conclusion in the tensor-network models from \cite{Pastawski:2015qua,Hayden:2016cfa}: in \cite{Pastawski:2015qua} it follows from eq. 4.8, since the code will only have complementary recovery if this inequality is saturated, and this means that the density matrix through the cut $\gamma_A$ is maximally mixed.  In \cite{Hayden:2016cfa} the triviality of the area operator follows from equation 5.9, which shows that the ``area term'' of the Renyi entropies is independent of $n$.

The origin of this trivial area operator is that we assumed the code subspace factorized into $\Ha\otimes \Hab$, and the only operators that can be shared between both factors are multiples of the identity.  To fix this, we need to generalize to a situation where the bulk algebras of operators in $\EA$ and $\EAb$ can have more in common.  This will clearly not be true if we continue to insist that they act on complementary factors of $\Hc$, so we will now drop this assumption and consider general operator algebras on $\Hc$.

\section{Operator-algebra quantum erasure correction}\label{opqecsec}
Operator-algebra quantum error correction is a generalization of subsystem quantum error correction introduced in \cite{beny2007generalization,beny2007quantum}.  The idea is to ask for recovery of only a subalgebra of the observables on $\Hc$.  In the special case where this subalgebra is the set of all operators on a tensor factor, this reduces to subsystem quantum error correction.  For the erasure channel it can be characterized by a theorem  generalizing theorems \eqref{cthm} and \eqref{sthm}, but before presenting and proving it we first need to recall some basic facts about subalgebras.

In this paper I will always take the subalgebra of interest to be a \textit{von Nuemann algebra on $\Hc$}.  This is a subset of the linear operators on $\Hc$ which is closed under addition, multiplication, hermitian conjugation, and which contains all scalar multiples of the identity (I will always assume that $\Hc$ is finite-dimensional, so there are no additional topological closure requirements). Von Neumann algebras are not particularly common in theoretical physics these days, and their general theory is quite sophisticated, especially in the infinite-dimensional case \cite{takesaki2003theory}.  The finite-dimensional case is more manageable, in appendix \ref{vnapp} I give a self-contained explanation of the basic results, including proofs.  I hope that it gives a relatively accessible entry to what can be a rather intimidating subject.  I will now state the essential results, so the appendix should only be necessary for readers who wish to understand the theory that underlies them.

The classification of von Neumann algebras on finite-dimensional Hilbert spaces, given by theorem \eqref{classthm}, tells us that for any von Neumann algebra $M$ on $\Hc$, we have a Hilbert space decomposition
\be\label{hilbertop}
\Hc=\oplus_\alpha \left(\Hh_{a_\alpha}\otimes \Hh_{\ab_\alpha}\right),
\ee
such that $M$ is just given by the set of all operators $\wt{O}$ that are block-diagonal in $\alpha$, and that within each block act as $\wt{O}_{a_\alpha}\otimes I_{\ab_\alpha}$, with $\wt{O}_{a_\alpha}$ an arbitrary linear operator on $\Hh_{a_\alpha}$.  In matrix form, we have
\be\label{Oform}
\wt{O}=
\begin{pmatrix}
\wt{O}_{a_1}\otimes I_{\ab_1} && 0 && \cdots\\
0 && \wt{O}_{a_2}\otimes I_{\ab_2} &&\cdots \\
\vdots && \vdots && \ddots
\end{pmatrix}
\ee
for any operator $\wt{O}\in M$.  The \textit{commutant of M}, denoted $M'$, and defined as the set of all operators on $\Hc$ that commute with everything in $M$, is also block-diagonal and consists of operators $\wt{O}'$ of the form
\be
\wt{O}'=
\begin{pmatrix}
I_{a_1}\otimes \wt{O}_{\ab_1}' && 0 && \cdots\\
0 && I_{a_2}\otimes \wt{O}_{\ab_2}' &&\cdots \\
\vdots && \vdots && \ddots
\end{pmatrix},
\ee
with the $\wt{O}_{\ab_\alpha}'$ arbitrary.  The \textit{center of M}, denoted $Z_M$, and defined as the operators in both $M$ and $M'$, consists of operators $\wt{\Lambda}$ of the form
\be\label{ocenter}
\wt{\Lambda}=
\begin{pmatrix}
\lambda_1\left(I_{a_1}\otimes I_{\ab_1}\right) && 0 && \cdots\\
0 && \lambda_2 \left(I_{a_2}\otimes I_{\ab_2}\right) &&\cdots \\
\vdots && \vdots && \ddots
\end{pmatrix},
\ee
with $\lambda_\alpha$ arbitrary elements of $\mathbb{C}$.  Thus we see that the blocks of the decomposition \eqref{hilbertop} arise from simultaneously diagonalizing all elements of $Z_M$. The special case where $M$ is the set of all operators on a tensor factor is realized if and only if $Z_M$ is trivial, in which case $M$ is called a \textit{factor}. 

In the following section it will be convenient to introduce orthonormal bases $|\wt{\alpha, i}\ran$ and $|\wt{\alpha, j}\ran$ for $\Hh_{a_\alpha}$ and $\Hh_{\ab_\alpha}$ respectively.  Together we can use these to build an orthonormal basis for $\Hc$:
\be\label{opbasis}
|\wt{\alpha, i j}\ran\equiv |\wt{\alpha, i}\ran\otimes|\wt{\alpha, j}\ran.
\ee  

Given a state $\wt{\rho}$ and a von Neumann algebra $M$ on $\Hc$, there is a definition of an entropy of $\wt{\rho}$ on $M$, which reduces to the standard von Neumann entropy when $M$ is a factor.  It is computed from the diagonal blocks $\wt{\rho}_{\alpha\alpha}$ of $\wt{\rho}$ in the following manner.  We first define
\be\label{prho}
p_\alpha \wt{\rho}_{a_\alpha}\equiv \Tr_{\ab_\alpha}\wt{\rho}_{\alpha\alpha},
\ee
with $p_\alpha\in [0,1]$ chosen so that $\Tr_{a_\alpha} \wt{\rho}_{a_\alpha}=1$.  This then implies that $\sum_\alpha p_\alpha=1$.  We then define
\be\label{entform}
S(\wt{\rho},M)\equiv -\sum_\alpha p_\alpha \log p_\alpha+\sum_\alpha p_\alpha S\left(\wt{\rho}_{a_\alpha}\right).
\ee
We can similarly define an entropy of $\wt{\rho}$ on $M'$, via
\be\label{prhop}
p_\alpha \wt{\rho}_{\ab_\alpha}\equiv \Tr_{a_\alpha}\wt{\rho}_{\alpha\alpha},
\ee
and
\be
S(\wt{\rho},M')\equiv -\sum_\alpha p_\alpha \log p_\alpha+\sum_\alpha p_\alpha S\left(\wt{\rho}_{\ab_\alpha}\right).
\ee
These entropies each consist of a ``classical'' piece, given by the Shannon entropy of the probability distribution $p_\alpha$ for the center $Z_M$, and a ``quantum'' piece given by the average of the von Neumann entropy of each block over this distribution.  The distribution $p_\alpha$ is shared between $M$ and $M'$.  The motivation for and properties of these entropies are discussed in more detail in section \ref{entapp} of the appendix.

\subsection{A theorem}
I can now present the basic theorem of operator-algebra quantum erasure correction \cite{beny2007generalization,beny2007quantum} (see also \cite{Almheiri:2014lwa} and \cite{hayden2004structure}):\footnote{\cite{hayden2004structure} is not explicitly about coding, but instead about the question of what sort of states saturate strong subadditivity, but Fernando Brandao has pointed out to me that many of their methods and results are quite similar to those I use and find here.  Perhaps there is a deeper connection at work?}
\begin{thm}\label{othm}
Say that $\Hh$ is a finite-dimensional Hilbert space, with a tensor product structure $\Hh=\HA\otimes\HAb$, and say that $\Hc$ is a subspace of $\Hh$ on which we have a von Neumann algebra $M$.  Moreover say that $|\wt{\alpha, i j}\ran$ is an orthonormal basis for $\Hc$ which is compatible with the decomposition \eqref{hilbertop} induced by $M$, as in \eqref{opbasis}, and that $|\phi\ran\equiv \frac{1}{\sqrt{|R|}} \sum_{\alpha,i,j} |\alpha,ij\ran_R|\wt{\alpha,ij}\ran_{A\ol{A}}$, where $R$ is an auxiliary system whose dimensionality is equivalent to that of $\Hc$. Then the following statements are equivalent:
\bi
\item[(1)] $\sum_\alpha |a_\alpha|\leq |A|$, we can decompose $\HA=\oplus_\alpha\left(\Hh_{A_1^\alpha}\otimes\Hh_{A_2^\alpha}\right)\oplus \Hh_{A_3}$ with $|A_1^\alpha|=|a_\alpha|$, and there exists a unitary transformation $U_A$ on $\HA$ and sets of orthonormal states $|\chi_{\alpha,j}\ran_{A_2^\alpha\ol{A}}\in \Hh_{A_2^\alpha\ol{A}}$ such that
\be
|\wt{\alpha,ij}\ran=U_A\left(|\alpha,i\ran_{A_1^\alpha}\otimes |\chi_{\alpha,j}\ran_{A_2^\alpha\ol{A}}\right).
\ee
Here $|\alpha,i\ran_{A_1^\alpha}$ is an orthonormal basis for $\Hh_{A_1^\alpha}$.

\item[(2)] For any operator $\wt{O}$ in $M$, there exists an operator $O_A$ on $\HA$ such that, for any state $|\wt{\psi}\ran\in \Hc$, we have
\begin{align}\nonumber
O_A|\wt{\psi}\ran&=\wt{O}|\wt{\psi}\ran\\
O_A^\dagger|\wt{\psi}\ran&=\wt{O}^\dagger|\wt{\psi}\ran.
\end{align}

\item[(3)] For any operator $X_{\ol{A}}$ on $\HAb$, we have
\be
P_{code}X_{\ol{A}}P_{code}=X'\Pc
\ee
with $X'$ some element of $M'$.  Here $P_{code}$ again denotes the projection onto $\Hc$. 
\item[(4)] For any operator $\wt{O}$ in $M$, we have
\be
[O_R,\rho_{R\Ab}(\phi)]=0.
\ee
Here $O_R$ is defined as the unique operator on $\Hh_R$ such that
\begin{align}
O_R|\phi\ran&=\wt{O}|\phi\ran\\
O_R^\dagger|\phi\ran&=\wt{O}^\dagger |\phi\ran,
\end{align}
explicitly it acts with the same matrix elements on $R$ as $\wt{O}^T$ does on $\Hc$.
\ei
\end{thm}
This theorem characterizes the ability of a code subspace to correct a subalgebra $M$ for the erasure of the physical degrees of freedom $\Ab$.  It reduces to theorem \eqref{sthm} if $M$ is a factor, and to theorem \eqref{cthm} if $M$ is all the operators on $\Hc$.  The equivalence of conditions (2), (3), and (4) was proven in appendix B of \cite{Almheiri:2014lwa}, I will give a more streamlined proof here that is closer to that already given for theorem \eqref{cthm}.  As far as I know condition (1) is new, it will be this condition that enables the connection to the RT formula in the following subsection.  
\begin{proof}

$(1)\Rightarrow (2)$: We can simply define $O_A\equiv U_A\left(\oplus_\alpha \left(O_{A_1^\alpha}\otimes I_{A_2^\alpha}\right)\right)U_A^\dagger$, where $O_{A_1^\alpha}$ acts on $\Hh_{A_1^\alpha}$ in the same way that $\wt{O}_{a_\alpha}$ from \eqref{Oform} acts on $\Hh_{a_\alpha}$.

$(2)\Rightarrow (3)$: Say that $\Pc X_{\Ab} \Pc=x'\Pc$, with $x'$ an operator on $\Hc$ but not an element of $M'$.  Then there must exist an $\wt{O}\in M$ and a state $|\wt{\psi}\ran\in \Hc$ such that $\lan \wt{\psi}|[x',\wt{O}]|\wt{\psi}\ran=\lan \wt{\psi}|[X_{\Ab},\wt{O}]|\wt{\psi}\ran\neq 0$, but such an $\wt{O}$ clearly cannot have an $O_A$, contradicting (2).

$(3)\Rightarrow (4)$: Say that $\wt{O}\in M$, and say that  $X_{\Ab}$ and $Y_R$ are arbitrary operators on $\HAb$ and $\Hh_R$ respectively.  We then have
$\Tr_{R\Ab}\left(O_R\,\rho_{R\Ab}(\phi)X_{\Ab} Y_R\right)=\lan\phi|X_{\Ab}Y_R O_R|\phi\ran=\lan\phi|X_{\Ab}Y_R \wt{O}|\phi\ran=\lan\phi|\wt{O}X_{\Ab}Y_R|\phi\ran=\lan\phi|O_RX_{\Ab}Y_R|\phi\ran=\Tr_{R\Ab}\left(\rho_{R\Ab}(\phi)O_RX_{\Ab} Y_R\right)$, which can only be true for arbitrary $X_{\Ab}$ and $Y_R$ if $[O_R,\rho_{R\Ab}(\phi)]=0$.

$(4)\Rightarrow (1)$: Our basis $|\alpha,ij\ran_R$ for $\Hh_R$ gives a decomposition 
\be
\Hh_{R\Ab}=\oplus_\alpha\left(\Hh_{R_\alpha}\otimes \Hh_{\ol{R}_\alpha}\otimes\HAb\right),
\ee
under which (4) implies that
\be
\rho_{R\Ab}(\phi)=\oplus_\alpha\left[\frac{|R_\alpha||\ol{R}_\alpha|}{|R|}\left(\frac{I_{R_\alpha}}{|R_\alpha|}\otimes \rho_{\ol{R}_\alpha \Ab}\right)\right]
\ee
for some $\rho_{\ol{R}_\alpha \Ab}$.  From $\rho_R=\frac{I_R}{|R|}$, we must have $\Tr_{\Ab} \,\rho_{\ol{R}_\alpha \Ab}=\frac{I_{\ol{R}_\alpha}}{|\ol{R}_\alpha|}$.  Since $\rho_{R\Ab}$ is purified by $|\phi\ran$, if we denote the rank of $\rho_{\ol{R}_\alpha \Ab}$ as $|\rho_{\ol{R}_\alpha \Ab}|$ then by the Schmidt decomposition it must be that
\be
\sum_\alpha |R_\alpha| |\rho_{\ol{R}_\alpha \Ab}|\leq |A|.
\ee
Therefore we can decompose 
\be
\HA=\oplus_\alpha\left(\Hh_{A_1^\alpha}\otimes \Hh_{A_2^\alpha}\right)\oplus \Hh_{A_3},
\ee
with $|A_1^\alpha|=|R_\alpha|=|a_\alpha|$ and $|A_2^\alpha|\geq |\rho_{\ol{R}_\alpha \Ab}|$.  For each $\alpha$ we can thus purify $\rho_{\ol{R}_\alpha \Ab}$ on $A_2^\alpha$, and from $\Tr_{\Ab} \,\rho_{\ol{R}_\alpha \Ab}=\frac{I_{\ol{R}_\alpha}}{|\ol{R}_\alpha|}$ this purification must have the form
\be
|\psi_\alpha\ran_{\ol{R}_\alpha A_2^\alpha \Ab}=\frac{1}{\sqrt{|\ol{R}_\alpha|}}\sum_j|\alpha,j\ran_{\ol{R}_\alpha}|\chi_{\alpha,j}\ran_{A_2^\alpha \Ab},
\ee
with the $|\chi_{\alpha,j}\ran$'s mutually orthonormal.  This then says we can purify $\rho_{R\Ab}$ as
\be
|\phi'\ran=\frac{1}{\sqrt{|R|}}\sum_{\alpha,ij}|\alpha,ij\ran_R|\alpha,i\ran_{A_1^\alpha}|\chi_{\alpha,j}\ran_{A_2^\alpha\Ab}.
\ee
Finally since $|\phi'\ran$ and $|\phi\ran$ are two purifications of $\rho_{R\Ab}$ on $A$, they must differ only by a unitary $U_A$, which implies (1).
\end{proof}
Since the last step of this proof is a bit complicated, it is worth mentioning that there is a simple proof \cite{Almheiri:2014lwa} that $(4)\Rightarrow(2)$: we observe that (4) implies that $O_R$ acts within the subspace of $\Hh_{R\Ab}$ that appears with nonzero coefficients in the Schmidt decomposition of $|\phi\ran$ into $R\Ab$ and $A$.  This then implies we can directly mirror $O_R$ back onto $A$, producing an $O_A$ that obeys (2).  

To apply this theorem to holography, we again need to introduce a version of the complementary recovery property, since we would also like to be able to represent operators in $\mathcal{E}_{\Ab}$ the entanglement wedge of $\Ab$ as operators on $\Ab$.  I will define a \textit{subalgebra code with complementary recovery} to be one where not only can we represent any element of $M$ on $A$ as in condition (2), we can also represent any element of $M'$ on $\Ab$.  The equivalence of (2) and (1) in theorem \eqref{othm} tells us that we then must have
\be\label{ocr}
|\wt{\alpha,ij}\ran=U_A U_{\Ab}\left(|\alpha,i\ran_{A_1^\alpha}|\alpha,j\ran_{\Ab_1^\alpha}|\chi_\alpha\ran_{A_2^\alpha \Ab_2^\alpha}\right).
\ee
Here we have introduced a decomposition $\HAb=\oplus_\alpha \left(\Hh_{\Ab_1^\alpha}\otimes \Hh_{\Ab_2^\alpha}\right)\oplus \Hh_{\Ab_3}$, with $|\Ab_1^\alpha|=|\ab_\alpha|$.  

Before proceeding, it seems appropriate to give a simple example of a subalgebra code with complementary recovery.  Consider the two-qubit system, with a code subspace $\Hc$ spanned by
\begin{align}\nonumber
|\wt{0}\ran&=|00\ran\\
|\wt{1}\ran&\equiv |11\ran.
\end{align}
The subalgebra $M$ I will consider is the one generated by $\wt{I}$ and $\wt{Z}$, with the latter acting as $\wt{Z}|\wt{0}\ran=|\wt{0}\ran$ and $\wt{Z}|\wt{1}\ran=-|\wt{1}\ran$.  This algebra is abelian, and thus has nontrivial center.  In fact center is all it has, so $|a_\alpha|=|\ab_\alpha|=1$, and $\alpha=0,1$.  Since $M=M'$, it must be that any operator in $M$ can be represented on either the first or the second physical qubit.  But this is clearly true, since $Z_1$ and $Z_2$ both act on $\Hc$ as $\wt{Z}$.

\subsection{A Ryu-Takayanagi formula}
Now let's consider an arbitrary encoded state $\wt{\rho}$ in a subalgebra code with complementary recovery on $A$ and $\Ab$.  From \eqref{prho},\eqref{prhop}, and \eqref{ocr}, we see that
\begin{align}\label{orho1}
\wt{\rho}_A&\equiv Tr_{\Ab}\wt{\rho}=U_A\Big(\oplus_\alpha \left(p_\alpha \rho_{A_1^\alpha}\otimes \chi_{A_2^\alpha}\right)\Big)U_A^\dagger\\\label{orho2}
\wt{\rho}_{\Ab}&\equiv Tr_{A}\wt{\rho}=U_{\Ab}\left(\oplus_\alpha \left(p_\alpha \rho_{\Ab_1^\alpha}\otimes \chi_{\Ab_2^\alpha}\right)\right)U_{\Ab}^\dagger,
\end{align}
where I've defined $\chi_{A_2^\alpha}\equiv \Tr_{\Ab_2^\alpha}|\chi_\alpha\ran\lan\chi_\alpha|$ and $\chi_{\Ab_2^\alpha}\equiv \Tr_{A_2^\alpha}|\chi_\alpha\ran\lan\chi_\alpha|$, and $\rho_{A_1^\alpha}$, $\rho_{\Ab_1^\alpha}$ act on $\Hh_{A_1^\alpha}$, $\Hh_{\Ab_1^\alpha}$ in the same way that $\wt{\rho}_{a_\alpha}$, $\wt{\rho}_{\ab_\alpha}$ do on $\Hh_{a_\alpha}$, $\Hh_{\ab_\alpha}$.  Finally if we define
\be\label{oarea}
\LA\equiv\oplus_\alpha S(\chi_{A_2^\alpha})I_{a_\alpha \ab_\alpha},
\ee
from \eqref{orho1}, \eqref{orho2} we find the Ryu-Takayanagi formulae:
\begin{align}\label{oflm1}
S(\wt{\rho}_A)&=\Tr \wt{\rho}\LA+S(\wt{\rho},M)\\\label{oflm2}
S(\wt{\rho}_{\Ab})&=\Tr \wt{\rho}\LA+S(\wt{\rho},M').
\end{align}
From \eqref{oarea} we see that the area operator $\LA$ is now nontrivial; $S(\chi_{A_2^\alpha})$ can take different values for different $\alpha$. Moreover we see that $\LA$ is of the form \eqref{ocenter}, and is thus an element of the center of $M$.

We can also study the relationships between the ``bulk'' and ``boundary'' modular Hamiltonians and relative entropies; the manipulations are similar to those for subsystem codes, and the result is that if we define modular Hamiltonians\footnote{See eqs. \eqref{rhoM4}, \eqref{orents} for motivation for this definition of $\wt{K}^\rho_M$.  I should really call it $\wt{\hat{K}}^\rho_M$, but the notational baggage is already getting ridiculous so I'll desist!}  
\begin{align}
\wt{K}^\rho_A&\equiv -\log \wt{\rho}_A\\
\wt{K}^\rho_{\Ab}&\equiv -\log \wt{\rho}_{\Ab}\\
\wt{K}^\rho_M&\equiv -\oplus_\alpha \left(\log(p_\alpha \wt{\rho}_{a_\alpha})\otimes I_{\ab_\alpha}\right)\\
\wt{K}^\rho_{M'}&\equiv -\oplus_\alpha \left(I_{a_\alpha}\otimes\log(p_\alpha \wt{\rho}_{\ab_\alpha})\right),
\end{align}
then we have
\begin{align}\label{ojlms}
\Pc \wt{K}^\rho_A\Pc&=\wt{K}^\rho_M+\LA\\\label{ojlms1}
\Pc \wt{K}^\rho_{\Ab} \Pc&=\wt{K}^\rho_{M'}+\LA\\\label{ojlms2}
S(\wt{\rho}_A|\wt{\sigma}_A)&=S(\wt{\rho}|\wt{\sigma},M)\\
S(\wt{\rho}_{\Ab}|\wt{\sigma}_{\Ab})&=S(\wt{\rho}|\wt{\sigma},M').\label{ojlms3}
\end{align}
Here the algebraic relative entropy $S(\wt{\rho}|\wt{\sigma},M)$ is defined by \eqref{rent}.  These are algebraic versions of the results of \cite{Jafferis:2015del}.

\subsection{An algebraic reconstruction theorem}
Before concluding, I will quickly point out that the reconstruction theorem of \cite{Dong:2016eik} can easily be extended to subalgebra codes with complementary recovery.  There it was shown that if $\Hh=\HA\otimes\HAb$, with $M$ is a factor algebra on $\Hc$, then the RT formulae \eqref{oflm1}, \eqref{oflm2} imply condition (3) of theorem \ref{othm}, and thus condition (2) (subregion duality in the entanglement wedge).  The argument goes through almost unmodified for general $M$, so I will proceed quickly.  

We first observe that there is an algebraic version of the  ``entanglement first law'', relating the modular Hamiltonian $\wt{K}^\rho_M$ and the algebraic entropy $S(\wt{\rho},M)$:
\be
S(\wt{\rho}+\delta\wt{\rho},M)=\Tr \left(\delta\wt{\rho}\wt{K}^\rho_M\right)+O\left(\delta\wt{\rho}^2\right).
\ee
Equating the linear terms on both sides of \eqref{oflm1} in a variation $\delta \wt{\rho}$ about a state $\wt{\sigma}$, we find
\be
\Tr\left(\delta \wt{\rho}_A\wt{K}^\sigma_A\right)=\Tr\left(\delta \wt{\rho}\left(\wt{K}_M^\sigma+\LA\right)\right).
\ee
Both sides of this equation are linear in $\delta\wt{\rho}$, so we can integrate to find
\be
\Tr\left(\wt{\rho}_A\wt{K}^\sigma_A\right)=\Tr\left(\wt{\rho}\left(\wt{K}_M^\sigma+\LA\right)\right).
\ee
This then implies equations \eqref{ojlms}, \eqref{ojlms2}, and an analogous argument for $\Ab$ implies equations \eqref{ojlms1}, \eqref{ojlms3}.

Now we will show condition (3), and its complementary version for $M'$, follow from \eqref{ojlms2}, \eqref{ojlms3}.  Consider a state $|\wt{\psi}\ran \in \Hc$, and operator $X_{\Ab}$ on $\HAb$, and an operator $\wt{O}\in M$.  Without loss of generality we can take $\wt{O}$ to be hermitian.  Now consider the quantity 
\be
\lan \wt{\psi}|e^{-i\lambda \wt{O}} X_{\Ab} e^{i\lambda \wt{O}}|\wt{\psi}\ran=\lan \wt{\psi}|e^{-i\lambda \wt{O}} \Pc X_{\Ab} \Pc e^{i\lambda \wt{O}}|\wt{\psi}\ran.
\ee
We will show that this is independent of $\lambda$, so in particular its linear variation with $\lambda$, proportional to $\lan \wt{\psi}|[\Pc X_{\Ab}\Pc,\wt{O}]|\wt{\psi}\ran$, must vanish for any $|\wt{\psi}\ran$.  This then implies condition (3) from theorem \ref{othm}.  Indeed notice that the states 
\be\label{thing}
|\wt{\psi}(\lambda)\ran\equiv e^{i\lambda \wt{O}}|\wt{\psi}\ran
\ee
have the property that the expectation value $\lan \wt{\psi}(\lambda)|\wt{O}'|\wt{\psi}(\lambda)\ran$ is independent of $\lambda$ for any $\wt{O}'\in M'$.  As explained below equation \eqref{rent}, this means that $S(\wt{\psi}(\lambda)|\wt{\psi}(\lambda'),M')=0$ for any $\lambda, \lambda'$.  From \eqref{ojlms3}, this then implies that $\Tr_A|\wt{\psi}(\lambda)|\ran\lan\wt{\psi}(\lambda)|$ is also independent of $\lambda$, which then implies the $\lambda$-independence of \eqref{thing}.  We can apply an identical argument exchanging $A\leftrightarrow\Ab$, $M \leftrightarrow M'$, so thus condition (3) holds in both cases and we thus have a subalgebra code with complementary recovery.

Combining this argument with theorem \ref{othm},  the RT formulae \eqref{oflm1}, \eqref{oflm2}, and the relative entropy results \eqref{ojlms2}, \eqref{ojlms3}, we at last arrive at the general reconstruction theorem \ref{bigrthm} quoted in the introduction. To review, the logic of the full proof is that subregion duality $\Rightarrow$ RT $\Rightarrow$ relative entropy equivalence $\Rightarrow$ subregion duality.

\section{Discussion}\label{discussion}
Having established the main technical results, we'll now see what they imply for the AdS/CFT correspondence.  

\subsection{Central elements and gauge constraints}\label{gaugesec}
\bfig
\includegraphics[height=1.5cm]{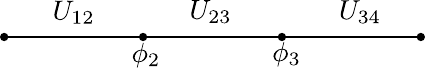}
\caption{Scalar lattice QED in 1+1 dimensions.  Each spatial link gets an element of $U(1)$, and each internal site gets a complex scalar.}\label{latticefig}
\efig
I'll first consider implications of the observation that the area operator $\LA$ must be in the center of the algebra $M$ associated to the entanglement wedge $\EA$.  We've seen that the presence of a nontrivial central operator indicates that $M$ is not a factor on the code subspace, which in bulk effective field theory is closely related to the presence of gauge symmetry \cite{Donnelly:2011hn,Casini:2013rba,Harlow:2014yoa,Radicevic:2015sza,Donnelly:2014fua,Donnelly:2015hxa,Harlow:2015lma,Donnelly:2015hta,Ma:2015xes,Soni:2015yga,Donnelly:2016auv,Donnelly:2016rvo}.  An easy way to illustrate this is in lattice scalar QED in $1+1$ dimensions, which we can study on four lattice sites arranged in a line.  The degrees of freedom are illustrated in figure \ref{latticefig}.  They have gauge transformations
\begin{align}
U_{i,i+1}'&=V_{i+1}U_{i,i+1} V_{i}^\dagger\\
\phi_i'&=V_i\phi_i,
\end{align}
and I'll impose boundary conditions where $V_1=V_4=1$ and $\phi_1=\phi_4=0$.  Gauge-invariant operators include
\begin{align}\nonumber
W&\equiv U_{12} U_{23}U_{34}\\\nonumber
E_{i,i+1}&\equiv - U_{i,i+1}\frac{\partial}{\partial U_{i,i+1}}\\\nonumber
\overleftarrow{\phi}_2&\equiv U_{12}^\dagger \phi_2\\
\overrightarrow{\phi}_3&\equiv \phi_3 U_{34}\\\nonumber
\overleftarrow{\pi}_2&\equiv U_{12} \pi_2\\\nonumber
\overrightarrow{\pi}_3&\equiv \pi_3 U_{34}^\dagger\\\nonumber
\rho_i&\equiv \frac{\partial}{\partial \phi_i} \phi_i-\phi_i^\dagger \frac{\partial}{\partial \phi_i^\dagger},
\end{align}
and the Gauss constraint can be written 
\be\label{gauss2}
E_{i,i+1}-E_{i-1,i}=\rho_i.
\ee
We can define an algebra $M_L$ of operators to the left of the link between sites two and three, which is generated by $\overleftarrow{\phi}_2$, $\overleftarrow{\pi}_2$, and $E_{12}$.  Its commutant $M_R\equiv M_L'$ is generated by $\overrightarrow{\phi}_3$, $\overrightarrow{\pi}_3$, and $E_{34}$.  $M_L$ has nontrivial center, since by the Gauss constraint we have $E_{23}=E_{12}+\rho_2=E_{34}-\rho_3$.  $E_{23}$ indeed is nontrivial, for example it doesn't commute with $W$, and in this example together with the identity it generates the entire center.  

Since $M_L$ has nontrivial center, if we wish to define the entropy of a state $\rho$ on $M$, we need to use eq. \eqref{entform}  \cite{Casini:2013rba}.  Indeed in \cite{Donnelly:2014fua,Donnelly:2015hxa} it was explained how correctly including this central contribution to the entropy from the electric fluxes through the entangling surface resolves an old discrepancy \cite{Kabat:1995eq} between replica-trick and direct Hilbert space calculations of the entropy of a region in Maxwell theory. In \cite{Donnelly:2014fua,Donnelly:2015hxa} these central electric degrees of freedom were called ``edge modes''.  

\bfig
\includegraphics[height=5cm]{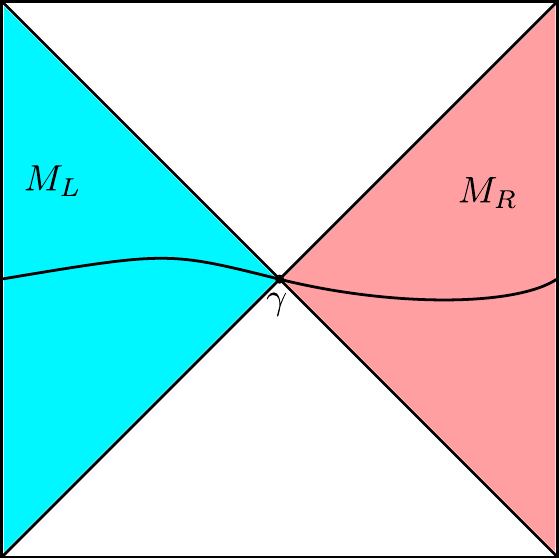}
\caption{An algebraic decomposition of the AdS-Schwarzschild geometry.  $M_L$ lives in the blue region, $M_R=M_L'$ lives in the red region, and the center corresponds to edge modes on the bifurcation surface $\gamma$.}\label{twoside}
\efig
Edge modes are especially interesting in the context of black holes and wormholes.  Indeed in \cite{Harlow:2014yoa}, the four-site QED example was used as a toy model of the maximally extended AdS-Schwarzschild geometry, as indicated in figure \ref{twoside}.  We can think of the algebra $M_L$ as corresponding to the degrees of freedom in the left exterior, and the degrees of freedom in $M_R=M_L'$ as living in the right exterior.  The edge modes live on the bifurcation surface $\gamma$, and correspond to integrating the normal electric field against an arbitrary function on that surface.  In this context these modes (and their gravitational counterparts) have recently been called ``soft hair'', by analogy with the asymptotic charges defined at spatial (or null) infinity \cite{Hawking:2016msc}.  This analogy can be misleading if taken too seriously, for example for AdS-Schwarzschild in greater than three spacetime dimensions, the asymptotic symmetry group is just the finite-dimensional conformal group (perhaps enhanced by a compact internal symmetry group such as $U(1)$), but a full set of horizon edge modes still exists.\footnote{Even in asymptotically-Minkowski situations, where there is a infinite-dimensional BMS group, most of the asymptotic charges are not involved in describing the process of black hole formation and evaporation, since they represent arbitrarily infrared excitations far away from the black hole.  During the black hole evaporation process, the amount of entropy produced per Schwarzschild time by the Hawking process is finite even in the limit $G\to 0$, while no gravitational asymptotic charges are excited in this limit since backreaction can be neglected.  So although the conservation of these asymptotic charges leads to some correlation in the Hawking radiation at finite $G$, it seems to be parametrically less than the amount which would be needed to purify the radiation.  For the simplest center-of-mass charges, where the correlation arises because the recoil of the black hole from emitting early radiation affects where it will be when it emits later radiation, this point was already made in \cite{Page:1979tc}.  Moreover even that correlation which is introduced does not seem like it should depend on the initial state of the black hole, so it is unclear to what extent this mechanism could restore information conservation even if it somehow restored purity of the final state.  By contrast the number of independent edge modes will be of order the horizon area in Planck units, although as we now discuss the precise number will be cutoff-dependent and cannot be computed within effective field theory.}

\bfig
\includegraphics[height=5cm]{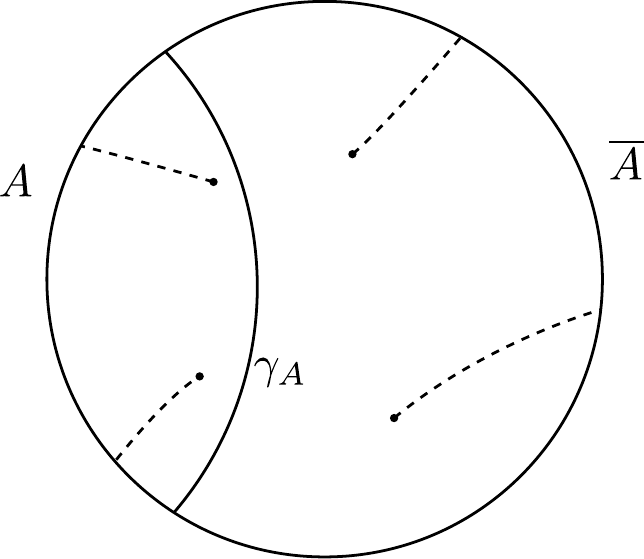}
\caption{Gravitational dressing in AdS.  Truly local operators do not exist in gravitational theories, but we can define pseudo-local operators by shooting geodesics from the boundary \cite{Heemskerk:2012np,Kabat:2013wga,Almheiri:2014lwa,Donnelly:2015hta,Donnelly:2015taa,Donnelly:2016rvo}. These operators will commute to all orders in perturbation theory with operators from which their entire geodesics are spacelike separated \cite{Almheiri:2014lwa,Donnelly:2015hta}, so provided that their geodesics lie entirely in $\mathcal{E}_A$ or $\mathcal{E}_{\Ab}$ then they will be in $M$ or $M'$ respectively, and they will commute with the area operator on the extremal surface (whose location is already defined gauge-invariantly without needing similar geodesics).  The area operator is thus in the center $M\cap M'$.}\label{dressing}
\efig
One important aspect of these edge modes is that any discussion of them is inherently UV-sensitive.  For example the gauge field could be emergent, in which case the true microscopic Hilbert space could still factorize.  In fact in \cite{Harlow:2015lma} it was pointed out that in the AdS/CFT correspondence, the microscopic description of the Hilbert space as two decoupled CFTs does indeed factorize, and this was used as evidence that we should think of any gauge fields in the bulk as emergent.  This conclusion is especially mysterious in the context of the RT formula, since the area operator is in the non-trivial center which arises because of bulk diffeomorphism invariance; it is the Noether charge of diffeomorphisms in the same way that the integrated electric flux is for electromagnetism \cite{Iyer:1994ys}.  I illustrate the central nature of the area operator in figure \ref{dressing}.  Since the factorization argument of \cite{Harlow:2015lma} implies that the gravitational constraints cannot really be viewed as holding in all states, gravity itself must also be emergent in a way that allows the Hilbert space to factorize.  So how can the RT formula hold with a nontrivial area operator if in fact the bulk algebra factorizes?  

The answer is that by working in a code subspace, we have chosen to restrict to states where the physics in the vicinity of $\gamma$ is described by bulk effective field theory.  In such states the microscropic degrees of freedom from which gravity emerges are fixed to be in a definite state, corresponding to the injection of $|\chi\ran$ in figure \ref{circuit} (or really in some combination of a small number of states given by the $|\chi_\alpha\ran$'s).  In the electromagnetic case we can have a situation where the gauge field emerges within effective field theory, such as the $\mathbb{CP}^{N-1}$ model considered in \cite{Harlow:2015lma}.  We may then extend the code subspace to include the fundamental charges from which the gauge field emerges, in which case the gauge-constraints become energetic rather than fundamental, so they do not pose any challenge for factorization.  It does not seem possible however for gravity to emerge within effective field theory \cite{Weinberg:1980kq,Marolf:2014yga}, so a code subspace that preserves gravitational effective field theory will never really be able to factorize, and we will always thus be able to have a nontrivial area operator.  

It is interesting to speculate about states outside of the code subspace, where the degrees of freedom from which the graviton emerges are liberated on either side of $\gamma$.  This sounds like a mechanism for making a firewall \cite{Almheiri:2012rt,Almheiri:2013hfa,Marolf:2013dba}, but note that this firewall would be at the edge of the entanglement wedge, \textit{not} at the horizon.  In general the entanglement wedge extends beyond the horizon \cite{Wall:2012uf,Headrick:2014cta}, and perhaps it usually goes far enough inside that its edge is not visible to infalling observers.  This would be a new kind of ``quantum cosmic censorship'', in which firewalls are generically present, but are typically far enough behind the horizon to be harmless.  Alternatively perhaps the entanglement wedge typically coincides with the causal wedge: if so, then firewalls are most likely here to stay. 

In any case, including all of the UV degrees of freedom in the code subspace just amounts to studying the full Hilbert space of the two CFTs, so the entropy of either side should just correspond to the bulk entropy on that side; the area term has disappeared. This is the ultimate realization of the standard observation that the separation of the right-hand-side of the RT formula into two terms is cutoff-dependent \cite{Solodukhin:2011gn}, or in our language code subspace-dependent.  In this limit the edge modes have fully dissolved into their microscopic constituents, which are finite in number due to the UV regulator provided by the CFT.  I'll say more about this in my discussion of the homology constraint below.

\subsection{Bit threads and multipartite entanglement}
Let's now consider in more detail the boundary interpretation of the RT formula suggested by fig. \ref{circuit}, or equivalently eq. \eqref{sstate} (or its algebraic generalization \eqref{ocr}).  From fig. \ref{circuit}, we see that for subsystem codes there is a flow of information from $A$ to $\Ab$, passing through the entangled state $|\chi\ran$.  The ``flux'' of this flow, given by the amount of entanglement in $|\chi\ran$, gives an irreducible contribution to the entanglement between $A$ and $\Ab$ for every state in the code subspace.  This contribution is quantified by the area terms in the RT formulae \eqref{sflm1}, \eqref{sflm2}.  For general subalgebra codes with complementary recovery this statement still basically holds, but we need to average over the center distribution $p_\alpha$ since there are multiple $|\chi_\alpha\ran$'s.

In fact the idea of interpreting the area piece of the RT formula via some kind of flow equations has appeared several times in the recent literature.  In \cite{Pastawski:2015qua} the max-flow, min-cut theorem was used to prove the RT formula in some tensor network models of holography, basically by manipulating the tensor network to extract fig. \ref{circuit}, although for simplicity the case with no bulk inputs (``holographic states'' as opposed to the ``holographic codes'' considered here) was considered.  In \cite{Freedman:2016zud}, a beautiful bulk rephrasing of the continuum RT formula was given which makes the connection to information flows essentially manifest.  In the remainder of this subsection I will explain in more detail the connection between fig. \ref{circuit} and the proposal of \cite{Freedman:2016zud}.

\bfig
\includegraphics[height=4cm]{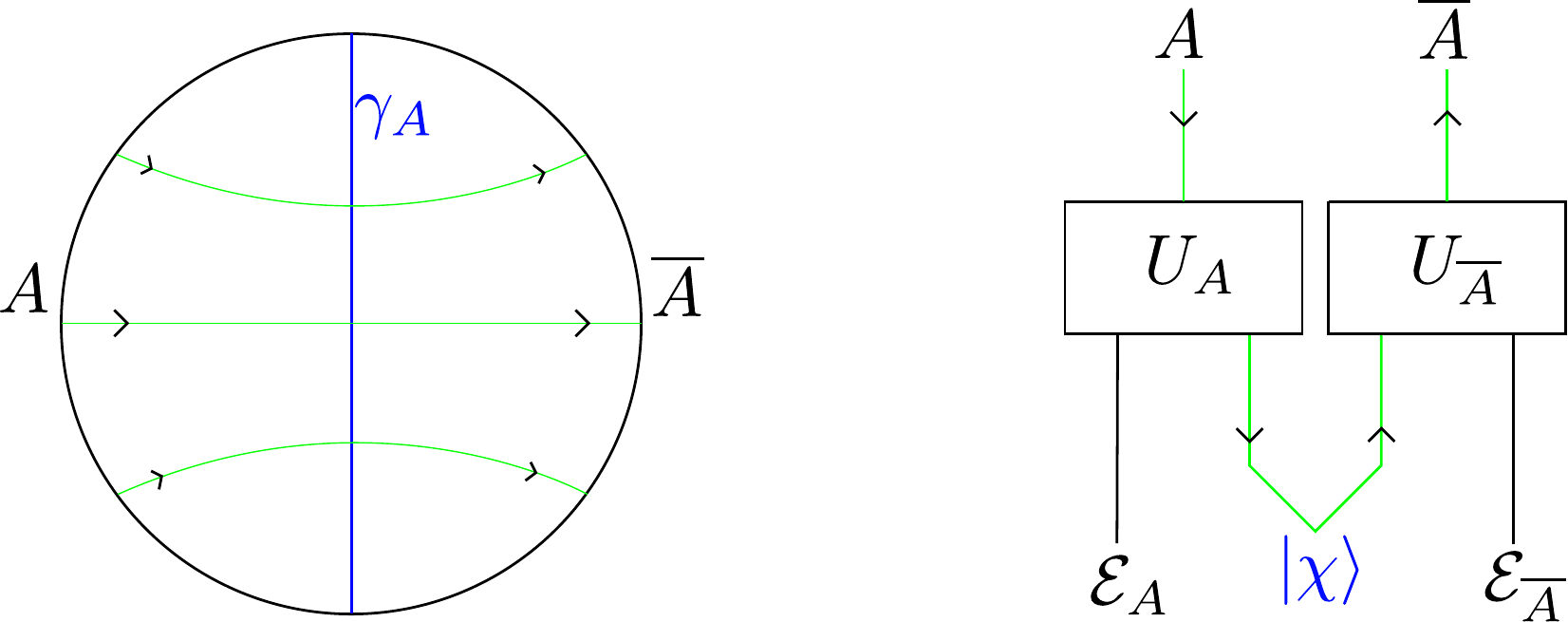}
\caption{Freedman-Headrick threads and the circuit interpretation of the RT formula.  The threads are shown in green in the left diagram; they are chosen to maximize the flux through $A$, and this maximal flux, determined by the bottleneck at $\gamma_A$, gives the entropy $S(\wt{\rho}_A)$.  In the circuit diagram these threads correspond to the information flux through the state $|\chi\ran$ that appears in eq. \eqref{sstate}.  We can thus interpret the Freedman-Headrick proposal as routing the circuit diagram through the bulk.}\label{threads}
\efig
The idea of \cite{Freedman:2016zud} is to consider smooth spatial vector fields $v(x)$ at a moment of time-reflection symmetry of the bulk,\footnote{There is also a covariant version of this proposal, which does not require this symmetry and that works in more or less in the same way \cite{headrickhubeny}.} which are divergenceless and have unit-bounded norm:
\begin{align}
\nabla\cdot v&=0\\
v\cdot v&\leq 1.
\end{align}
We then look for a $v(x)$ which maximizes the flux $\int_A *v$.  Naively it may seem like we could simply arrange the maximal flux to be given by the area of $A$, but this is not the case.  The reason is that $\int_A *v=\int_\gamma *v$, where $\gamma$ is any (spacetime codimension two) surface in the bulk which is homologous to $A$, and it might well be that the area of $\gamma$ is less than that of $A$.  Indeed we can at best arrange for $\int_A *v=\int_{\gamma_A}*v$, where $\gamma_A$ is the minimal-area surface homologous to $A$, and in fact a continuous version of max-flow, min-cut ensures that we can attain this for some $v(x)$ \cite{Freedman:2016zud}.  The proposal is then that we re-interpret the RT formula as saying that\footnote{For now we are assuming that the bulk entropy piece is subleading in $G$ and can be neglected.}
\be
S(\wt{\rho}_A)=\frac{1}{4G}\mathrm{Max}_v \int_A *v.
\ee
The flow lines of a $v(x)$ which attains this maximum are interpreted as giving a density of ``bit threads'', which graphically illustrate the entanglement between $A$ and $\Ab$.  What we learn from figure \ref{circuit} is that this is more than an analogy, it is \textit{actually} how the RT formula is realized from the boundary point of view. I indicate this in figure \ref{threads}. Finding a maximal $v(x)$ corresponds to applying unitaries to $A$ and $\Ab$ to distill the maximal amount of entanglement between $A$ and $\Ab$.  It may seem that a bit thread configuration $v(x)$ contains more information than fig. \ref{circuit}, but the various conditions imposed on $v(x)$, together with the large non-uniqueness of the maximal $v(x)$, mean that the essential information is the same.

\bfig
\includegraphics[height=4cm]{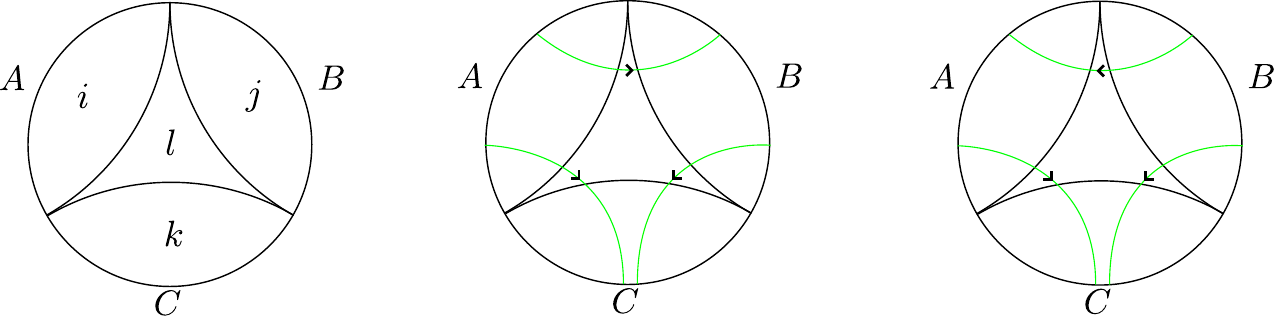}
\caption{Tripartite entanglement.  On the left, I indicate the locations of the bulk degrees of freedom from eq. \eqref{triple}.  In the center I draw threads  $v_{A,B}$ that simultaneously maximize the flux through $A$ and $AB$, while on the right I draw threads $v_{B,A}$ that instead maximize the flux through $B$ and $AB$.}\label{mp}
\efig
So far I have focused on bipartite entanglement between $A$ and $\Ab$, but it is also interesting to consider multipartite decompositions, such as the one shown in the right diagram of figure \ref{subregions}.  For simplicity I will only consider the subsystem code case, where we take the bulk algebra to factorize into different spatial regions.  For the tripartite decomposition of figure \ref{subregions} there are four interesting bulk regions, labeled in the left diagram of figure \ref{mp}. $|\wt{i}\ran$, $|\wt{j}\ran$, and $|\wt{k}\ran$ denote complete bases for the bulk degrees of freedom in $\mathcal{E}_A$, $\mathcal{E}_B$, and $\mathcal{E}_C$, and $|\wt{l}\ran$ is a complete basis for the remaining bulk degrees of freedom, which are simultaneously in $\mathcal{E}_{AB}$, $\mathcal{E}_{AC}$, and $\mathcal{E}_{BC}$.  If we assume entanglement wedge reconstruction holds for all entanglement wedges, then an argument similar to that for \eqref{sstate} tells us that we must have decompositions $\HA=\Hh_{A_1}\otimes \Hh_{A_2}\oplus\Hh_{A_3}$, $\Hh_B=\Hh_{B_1}\otimes \Hh_{B_2}\oplus\Hh_{B_3}$, and $\Hh_C=\Hh_{C_1}\otimes \Hh_{C_2}\oplus\Hh_{C_3}$, unitaries $U_A$, $U_B$, and $U_C$, and a set of orthonormal states $|\chi_l\ran_{A_2B_2C_2}$ such that
\be\label{triple}
|\wt{ijkl}\ran=U_A U_B U_C \Big(|i\ran_{A_1}|j\ran_{B_1}|k\ran_{C_1} |\chi_l\ran_{A_2B_2C_2}\Big).
\ee
Moreover the states $|\chi_l\ran$ must define a code subspace of $\Hh_{A_2B_2C_2}$ which gives a conventional quantum error correcting code that can recover the $l$- information on any two of $A_2$, $B_2$, or $C_2$.  Eq. \eqref{triple} is the tripartite version of figure \ref{circuit}.  To compare with the bit threads of Freedman and Headrick, we will again assume that our code subspace is small enough that the leading-order pieces of the boundary von Neumann entropies comes from a fixed state $|\chi\ran_{A_2B_2C_2}\ran$, which remains after we have decoded $l$ onto a subfactor of our choice.  Since $A_1$, $B_1$, and $C_1$ contribute only subleadingly to the entropies, for the rest of this section I will ignore them typographically and just consider the entanglement structure of a single tripartite state $|\chi\ran_{ABC}$.  

When $|\chi\ran$ was a state in a bipartite Hilbert space, it was easy to classify its entanglement structure by way of the Schmidt decomposition. Indeed for any state $|\chi\ran_{A\Ab}$ there are orthornomal states $|n\ran_A$, $|n\ran_{\Ab}$ such that
\be
|\chi\ran_{AB}=\sum_n \sqrt{p_n} |n\ran_A |n\ran_{\Ab},
\ee
with  $p_n>0$ and $\sum_n p_n=1$.  The bit threads simply run from $A$ to $\Ab$, with a flux given by $S(\chi_A)=S(\chi_{\Ab})=-\sum_n p_n \log p_n$.  Unfortunately there is no tripartite version of the Schmidt decomposition, so we need to do something less precise.  We can begin by Schmidt decomposing $|\chi\ran$ into $AB$ and $C$, again with with $S(\chi_{AB})=S(\chi_{C})$, but now we need to make sense of the mixed state $\chi_{AB}$.  Freedman and Headrick showed that it is possible to find a set of threads $v_{A,B}$ that simultaneously maximize the flux through $A$ and $AB$, a set of threads $v_{B,A}$ that simultaneously maximize the flux through $B$ and $AB$, but that it is not in general possible to maximize the flux through $A$, $B$, and $AB$ simultaneously.  They then characterized the multipartite entanglement of $\chi_{AB}$ by how the threads move as we switch from $v_{A,B}$ and $v_{B,A}$.  They argued that threads from $A$ to $B$ which switch direction correspond to bipartite entanglement between $A$ and $B$, that threads from $A$ (or $B$) to $C$ which do not move correspond to bipartite entanglement between $A$ (or $B$) and $C$, and that threads from $A$ to $C$ which switch to threads from $B$ to $C$ correspond to GHZ-type entanglement between $A$, $B$, and $C$.   The first two cases are illustrated in the center and right diagrams of figure \ref{triple}.  

Unfortunately it is not true that an arbitrary state on $ABC$ can be written up to unitaries on $A$, $B$, and $C$ as a tensor product of GHZ and bipartite states.  For example in the three qubit system, the state 
\be
|\phi\ran_{ABC}=\frac{1}{\sqrt{2}}|000\ran+\frac{1}{2}|101\ran+\frac{1}{2}|011\ran
\ee
cannot be factorized into a bipartite entangled state on two qubits and a pure state on a third, and it is not GHZ since $S(\psi_A)\neq S(\psi_{AB})$.  In general the full entanglement structure of the state $|\chi\ran_{ABC}$ will be more sophisticated than what can be captured just by the thread picture.  Nonetheless the threads $v_{A,B}$ and $v_{B,A}$ do exist, so they have to mean something.  I propose that we can interpret them as representing the fact that for any state $|\chi\ran_{ABC}$, we can find a pure state $|\psi\ran_{ABC}$ which \textit{is} just a tensor product of bipartite states between the various factors, and whose von Neumann entropies on $A$, $B$, and $C$ agree with those of $|\chi\ran_{ABC}$.  This state will \textit{not} in general obey $|\psi\ran_{ABC}=|\chi\ran_{ABC}$ up to unitaries on $A$, $B$, and $C$, but we will just have to live with that.  To see that such a state always exists, note that if we have
\be
|\psi\ran_{ABC}=|\psi^{AB}\ran_{A_1B_1}\otimes |\psi^{AC}\ran_{A_2 C_1}|\psi^{BC}\ran_{B_2 C_2},
\ee
where we have split $A$, $B$, $C$ into factors $A_1$, $A_2$, etc, then we can choose these factor states so that
\begin{align}
S(\Tr_{B_1}\psi^{AB})&=\frac{1}{2}\left(S(\chi_{A})+S(\chi_{B})-S(\chi_{AB})\right)\\
S(\Tr_{C_1}\psi^{AC})&=\frac{1}{2}\left(S(\chi_{A})-S(\chi_{B})+S(\chi_{AB})\right)\\
S(\Tr_{C_2}\psi^{BC})&=\frac{1}{2}\left(-S(\chi_{A})+S(\chi_{B})+S(\chi_{AB})\right).
\end{align}
These entropies are positive by the positivity of mutual information $S_A+S_B-S_{AB}\geq 0$ and the Araki-Lieb inequality $|S_A-S_B|\leq S_{AB}$, so we can always find states that attain them.\footnote{If the Hilbert spaces are finite-dimensional there may not be enough room in $A$, $B$, $C$ to make these choices, but in AdS/CFT the relevant Hilbert spaces are infinite-dimensional so there is always enough room.}  I thus claim that we should view the bit threads for $|\chi\ran_{ABC}$ as representing the bipartite entanglement in $|\psi\ran_{ABC}$, with the directions set by whether we are considering $v_{A,B}$, $v_{B,A}$, $v_{A,C}$, etc.  This proposal does not seem totally satisfactory, for example the state $|\psi\ran_{ABC}$ will not necessarily compute the correct Renyi entropies for the various regions, but then we don't know how to compute those from the threads either.

Interestingly we did not need to use GHZ-type states in $|\psi\ran_{ABC}$, although they do have a thread description.  Perhaps considering more regions will require them.  Since we know that the thread prescription is equivalent to the area term of the RT formula, once we consider four regions the entropies will obey inequalities such as the monogamy of mutual information that are not actually true for general quantum states \cite{Hayden:2011ag,Bao:2015bfa}, so at that point we will start seeing restrictions on which entropies can be represented by threads.   

Figure \ref{circuit} and expressions \eqref{sstate}, \eqref{triple} should make it straightforward to extend the Freedman-Headrick picture to include the bulk-entropy piece of the RT formula, but I won't work this out here.

\subsection{Linearity and the homology constraint}
The original Ryu-Takayanagi formula \cite{Ryu:2006bv}, \cite{Hubeny:2007xt,Lewkowycz:2013nqa} did not contain the bulk entropy term in \eqref{flm}, it simply said that 
\be\label{ort}
S(\wt{\rho}_A)=\Tr \wt{\rho}\LA,
\ee
with 
\be
\LA=\frac{\mathrm{Area}(\gamma_A)}{4G}.
\ee
Here $\gamma_A$ is an extremal-area codimension-two surface homologous to $A$, where homologous means that $A\cup \gamma_A=\partial \Xi$, with $\Xi$ some codimension-one spacelike submanifold with boundary in the bulk \cite{Headrick:2007km,Haehl:2014zoa}.  If there is more than one such $\gamma_A$, we choose the one of minimal area.  It is immediately clear that there can be no code subspace where eq. \eqref{ort} holds precisely for arbitrary $\wt{\rho}$, since the right hand side is linear in $\wt{\rho}$ but the left hand side is not.  As far as I know this issue was first discussed in detail in \cite{Papadodimas:2015jra}, where it was used as justification for more general violations of the linearity of quantum mechanics in a proposed description of the interior of black holes (see also \cite{Harlow:2014yoa,Marolf:2015dia} for more on this proposal, and \cite{Raju:2016vsu} for an attempt to reconcile it with quantum mechanics).  Quite recently \cite{Almheiri:2016blp} appeared, which extensively explored the nonlinearity of \eqref{ort}, and in particular which gave two explicit situations where it leads to a breakdown of eq. \eqref{ort}.  In this section I will argue that, once the bulk entropy term is restored to \eqref{ort}, as in \eqref{flm}, then there no longer need be any tension with the RT formula holding throughout a code subspace.  Indeed this must have been the case, since throughout the paper we have discussed examples, such as the three qutrit code or the tensor networks of \cite{Pastawski:2015qua,Hayden:2016cfa}, where the RT formula provably holds in a nontrivial subspace.  

I'll first consider the behavior of the RT formula in admixtures of states with distinct classical geometries \cite{Papadodimas:2015jra,Almheiri:2016blp}:
\be\label{adstate}
\wt{\rho}=\sum_i p_i \wt{\rho}_i.
\ee
The idea is that the $\wt{\rho}_i$'s here are coherent states of $\LA$, with only exponentially small overlaps.  We then have
\be\label{mixent}
S(\wt{\rho}_A)=-\Tr\left(\sum_i p_i \wt{\rho}_i\log\left(\sum_i p_i \wt{\rho}_i\right)\right)\approx -\sum_i p_i \log p_i+\sum_i p_i S\left(\wt{\rho}_{i,A}\right).
\ee
The first term on the right hand side is called \textit{entropy of mixing}, and it is a manifestation of the nonlinearity of the entropy.  In particular it would not arise if \eqref{ort} applied for all the $\wt{\rho}_i$ as well as $\wt{\rho}$.  In \cite{Papadodimas:2015jra,Almheiri:2016blp} it was argued that, since this term is subleading in $G$, if we do not consider exponentially many $\wt{\rho}_i$'s, it does not really pose a challenge to \eqref{ort}.  But I'll now argue that something better is true: this term is actually \textit{accounted for} by the bulk entropy term in the full RT formula \eqref{flm}.  The argument is easy: states with different values for $\LA$ necessarily lie in different blocks of the central decomposition \eqref{hilbertop}.  So the $p_i$'s in \eqref{adstate} are a subset of the $p_\alpha$'s in \eqref{prho}, and the entropy of mixing then obviously arises from the ``classical'' term in the bulk algebraic entropy \eqref{entform}.  The ``quantum'' term in \eqref{entform} accounts for the bulk contributions to the entropies in the second term of the right hand side of \eqref{mixent}, and the area terms match trivially by linearity.  So entropy of mixing is no obstruction to the RT formula \eqref{flm} holding exactly in a subspace that includes states with classically different geometries.

\bfig
\includegraphics[height=3cm]{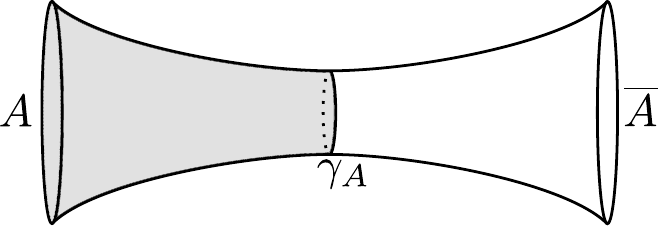}
\caption{Homology and the RT formula for the thermofield double state.  Insisting that $\gamma_A$ is homologous to $A$ prevents us from taking $\gamma_A$ to be empty, as would otherwise be allowed since $\partial A=0$.  We instead must take $\gamma_A$ to be the birfucation point. This diagram represents a time-reversal-symmetric slice of the geometry of figure \ref{twoside}.  The surface $\Xi$ is shaded grey.}\label{wormhole}
\efig
In \cite{Almheiri:2016blp}, it was also pointed out that a more subtle problem in the validity of \eqref{ort} arises when we attempt to include black holes into the code subspace.  Let's first recall the standard story for how to think about the thermofield double state of two CFTs, 
\be\label{TFD}
|TFD\ran=\frac{1}{\sqrt{Z}}\sum_i e^{-\beta E_i/2}|i^*\ran_L|i\ran_R,
\ee
in the situation where $\beta$ is small enough that in the bulk we expect this to be described by the AdS-Schwarzschild geometry shown in figure \ref{twoside}.  We can take our region $A$ to be the entire left CFT, in which case we have the situation of figure \ref{wormhole}.  Since the left CFT is in the thermal state $\frac{1}{Z} e^{-\beta H}$, its entropy is nonzero; to leading order in $G$ it is given by $\frac{\mathrm{Area}(\gamma_A)}{4G}$.  This suggests that if we consider just a single CFT with a black hole in a thermal state, we should think of the surface $\gamma_A$ as being located at the horizon.  

The tension with linearity pointed out in \cite{Almheiri:2016blp} arises if we addionally consider a complete set of single-CFT black hole microstates $|i\ran$ in some energy band of sufficiently high energy that black holes are stable, which we can take to be energy eigenstates as in \eqref{TFD}.  By the eigenstate thermalization hypothesis, we expect that the geometry outside of the horizon of these states to be close to that of the AdS-Schwarzschild geometry, but in fact the von Neumann entropy of the CFT in any particular microstate will be zero since the state is pure.  So if we believed \eqref{ort} held in all microstates, then by linearity we would conclude that the area operator $\LA$, with $A$ taken to be the entire boundary, must be zero on the subspace of the Hilbert space spanned by these microstates, and thus that $\gamma_A$ must be empty.  But this would contradict the nonvanishing of this operator in the thermal state, which is an admixture of these microstates but where $\gamma_A$ lies on the horizon.  For this reason, the authors of \cite{Almheiri:2016blp} identified the homology constraint as the origin of the linearity problem in the RT formula, since it apparently applies in the mixed thermal state, but not in pure microstates.  

Indeed pure state black holes have always been somewhat awkward to fit into discussions of the RT formula \eqref{ort}.  The standard excuse is that if a pure state black hole is created by the formation of a shell of matter, then the homology constraint does not prevent us from sliding the surface $\gamma_A$ down under the collapse and then contracting it to zero size.  Unfortunately most pure microstates do not correspond to black holes that formed all at once, and without a general understanding of what the geometry behind their horizons is, application of the homology constraint is ambiguous at best.  Moreover what if the matter shell is mixed?  For example we could consider collapsing two entangled matter shells to form two entangled black holes in the TFD state \cite{Susskind:2014yaa}. Prior to the collapse, the conventional understanding of the RT formula for an entanglement wedge containing only one of the shells would include the entropy of its shell in the bulk entropy term, while after the collapse it would come from the area term.  Why should we treat this entropy differently before and after the collapse \cite{Susskind:2014yaa}?

\bfig
\includegraphics[height=6.5cm]{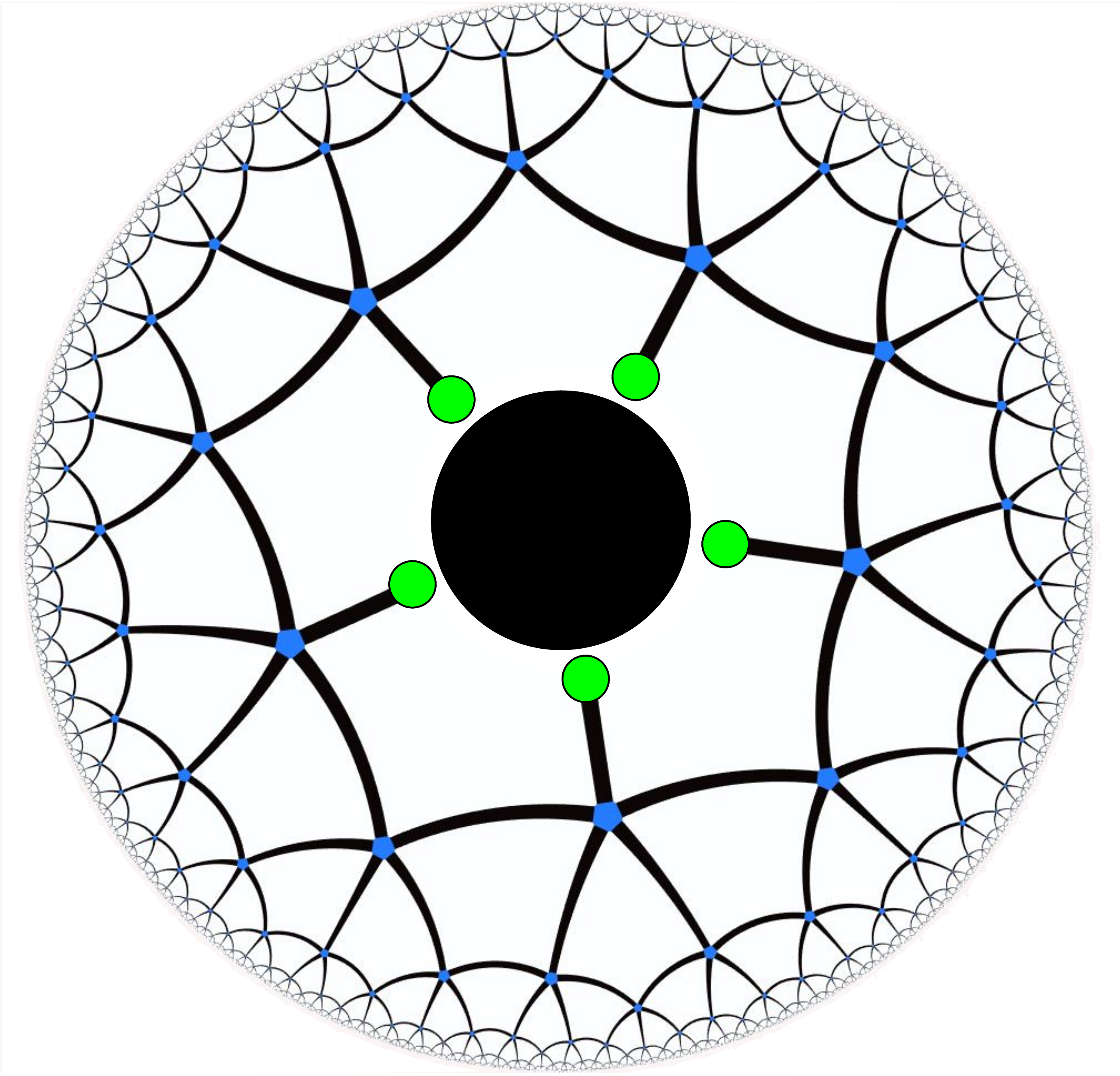}
\caption{A black hole in the tensor network model of \cite{Pastawski:2015qua}. The network gives an isometry from the green``microstate'' legs previously attached to the removed tensor(s), together with the ``bulk field'' legs attached to the blue tensors outside of the black hole, to the boundary ``CFT'' legs.}\label{blackhole2}
\efig
One possible resolution of all this would be to avoid considering a code subspace that contains all of the microstates in a fixed energy band, but this is somewhat unsatisfying, especially since in section six of \cite{Pastawski:2015qua} it was explained how subregion duality is possible in a tensor network model even if the code subspace includes all the microstates of a black hole of some fixed energy.  In that model, black hole microstates are produced by removing tensors from the network wherever the black holes are located, as illustrated in figure \ref{blackhole2}.  The isometric nature of the network implies that the entropy of the full boundary will be given by the entropy of whatever state is fed into the green microstate legs and the blue bulk field legs.  So apparently the RT formula \eqref{flm} still holds for arbitrary states fed into these legs, \textit{provided that we view the black hole entropy as contributing to the bulk entropy term rather than the area term}.  In the remainder of this section I will explore the consequences of this idea, which I claim removes any remaining tension between linearity the RT formula \eqref{flm}.

\bfig
\includegraphics[height=4.5cm]{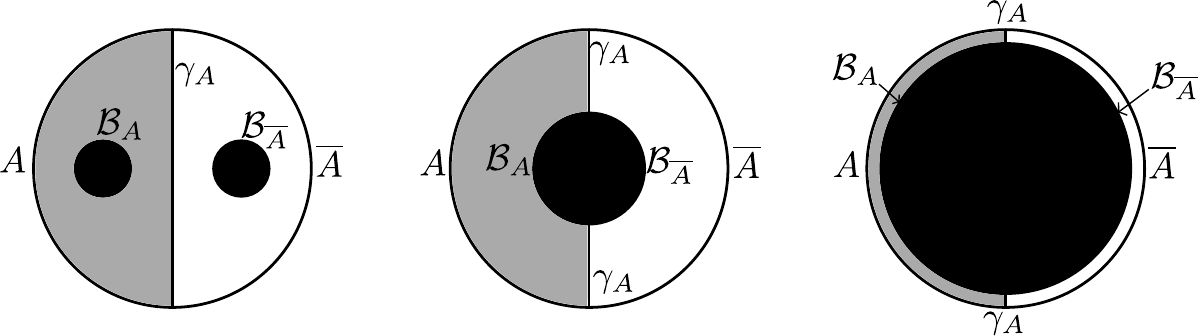}
\caption{Area and bulk terms in the RT formula according to proposition \ref{bhprop}, for various configurations of black holes.  On the left we have a black hole in each entanglement wedge, each of which contributes its von Neumann entropy to the bulk entropy term in the RT formulae for $S(\wt{\rho}_A)$ and $S(\wt{\rho}_{\Ab})$ respectively.  In the center we have a single black hole, whose degrees of freedom partly contribute to $S(\wt{\rho}_A)$ and partly contribute to $S(\wt{\rho}_{\Ab})$, again through the bulk entropy terms. On the right we take the limit where the black hole fills the entire space, in which case the area term is eventually removed entirely.  In each case the surface $\Xi$ is shaded grey, and black holes are black.}\label{bhfig}
\efig
Let's first recall that in subsection \ref{gaugesec}, we have already seen that the decomposition of the CFT entropy of a region into an area piece and a bulk entropy piece is UV-sensitive.  By enlarging the code subspace to allow more UV degrees of freedom to vary, we can move entropy from the area piece to the bulk entropy piece.  We can think of my proposal to view black hole entropy as bulk entropy in this context: sometimes the code subspace is small enough that we can get away with including black hole entropy in the area piece and applying the homology constraint (for example studying only small perturbations of the TFD), but sometimes we can't.  The rule which works in general is to always include it in the bulk entropy piece.  I thus offer the following proposition:\footnote{This proposition needs to be better-formulated to really apply in general time-dependent situations, but it will be good enough for my examples.  I also am not sure how to deal with changes of $\gamma_A$ which increase its area but decrease the horizon part of the entropy by more, one guess is that any intersections between $\gamma_A$ and $\mathcal{B}_A$ are located by extremizing the sum of the area and bulk entropy terms, as suggested by \cite{Engelhardt:2014gca}, but I'm not sure if this is correct.}
\begin{prop}\label{bhprop}
Say we are given a CFT subregion $A$.  The correct codimension-two surface $\gamma_A$ to use in the RT formula \eqref{flm} is an extremal-area surface such that $\partial \Xi=A \cup \gamma_A \cup \mathcal{B}_A$, with $\Xi$ a codimension-one submanifold with boundary, and $\mathcal{B}_A$ some codimension-two piece of any horizons that might be around.  The bulk entropy in the RT formula should then include a contribution from any effective field theory degrees of freedom in $\Xi$, as well as any horizon degrees of freedom in $\mathcal{B}_A$.  
\end{prop}
\bfig
\includegraphics[height=5.5cm]{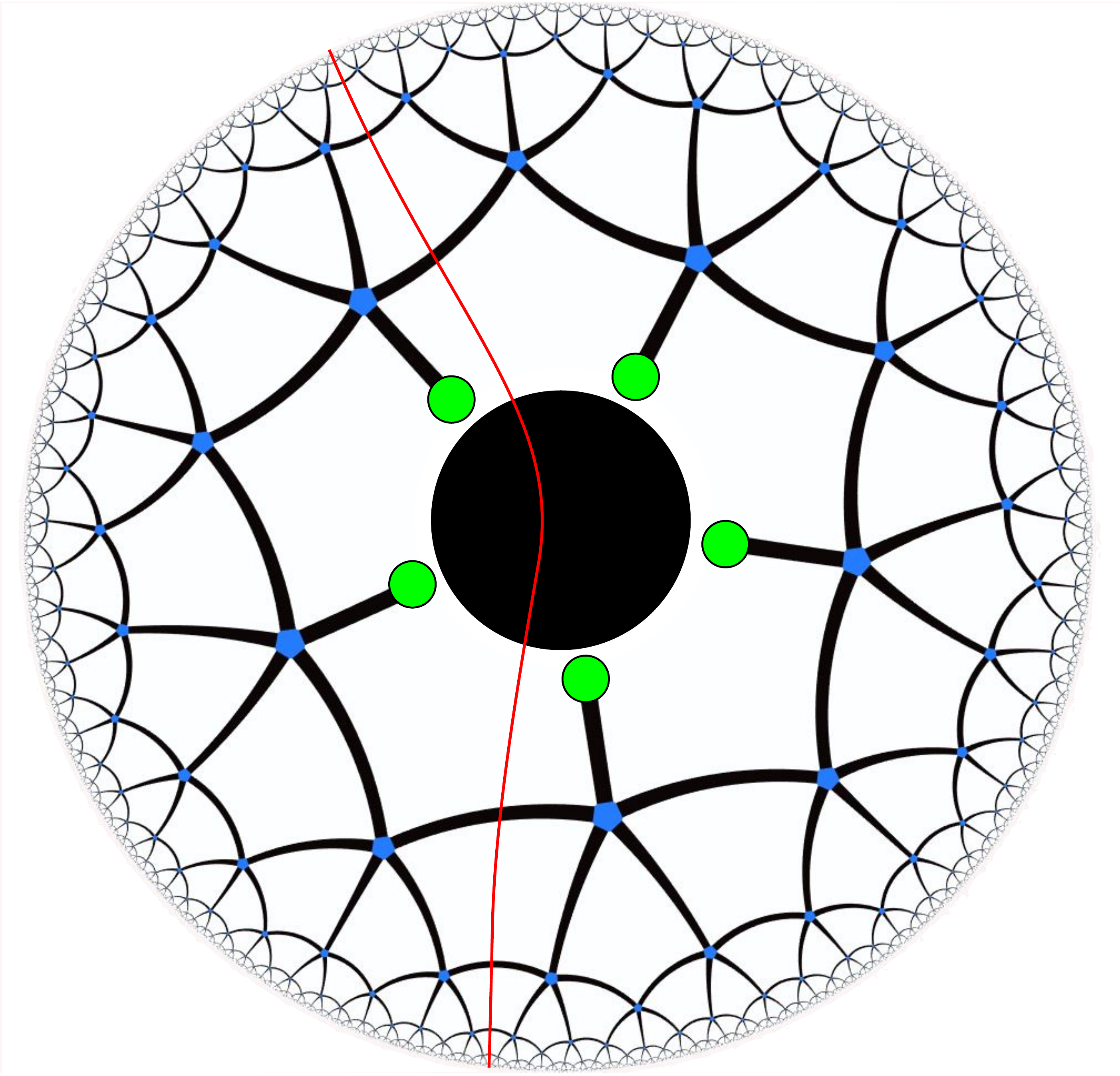}
\caption{A tensor network cut that divides black hole microstates. The area term in the RT formula comes from the links cut by the red line, while green microstate legs and blue bulk field legs each contribute to the bulk entropy term for their respective side.  The network on either side is an isometry from the cut legs, microstate legs, and the bulk field legs to the boundary subregion on that side of the cut}\label{cutbh}
\efig
We can think of $\mathcal{B_A}$ as the pieces of black hole horizon that lie within the entanglement wedge $\EA$.  Some examples illustrating this rule for arbitrary black hole microstates, pure or mixed, are given in figure \ref{bhfig}.  In each case, proposition \ref{bhprop} can be confirmed in the tensor network black holes of \cite{Pastawski:2015qua} (or an analogous construction using random tensors as in \cite{Hayden:2016cfa}): a concrete example is shown in figure \ref{cutbh}.
From figure \ref{bhfig}, it is clear that including the microstates of larger and larger the black holes allows fewer and fewer bulk operators to be encoded redundantly, and eventually we are just left with the full Hilbert space of the CFT and no remaining redundancy.  This is in keeping with the general picture of holography advocated in \cite{Almheiri:2014lwa}.

Although proposition \ref{bhprop} thus can explain the validity of the RT formula for rather permissive code subspaces, it has the downside that we have essentially removed the black hole interior from the discussion by fiat.  This is to be contrasted with the approach of \cite{Susskind:2014yaa}, which instead tries to move the bulk entropy contribution to the RT formula into the area piece, therefore geometrizing even the entanglement of the ordinary bulk fields via a kind of ``quantum homology constraint''.  This approach seems more natural from the point of view of ``ER=EPR'' \cite{Maldacena:2001kr,Swingle:2009bg,VanRaamsdonk:2010pw,Hartman:2013qma,Maldacena:2013xja}, but it seems like it cannot be consistent with linearity unless we consider only rather small code subspaces.  Should we therefore conclude that linearity requires most black hole microstates to not have interiors?  This is more or less the firewall argument \cite{Almheiri:2012rt,Almheiri:2013hfa,Marolf:2013dba},  but so far this conclusion seems premature.  Naively proposition \ref{bhprop} would suggest defining the entanglement wedge $\EA$ as the bulk domain of dependence of $\Xi$, which by construction never goes behind the black hole horizons.  But in fact at least in some states we know it can be defined to go further by trading some of the microstate degrees of freedom for interior bulk degrees of freedom, and even in generic states we may yet be able to extend it somewhat beyond the horizon.  Perhaps this requires nonlinear violations of quantum mechanics, as advocated in \cite{Papadodimas:2015jra}, but perhaps not.  I hope to return to this in the future. 

\subsection{Limitations}
I'll close by discussing a few points where my analysis clearly needs to be improved from the point of view of applying it to holography.

First of all, theorem \ref{bigrthm} gives an equivalence between three seemingly different properties of a subspace $\Hc\subset \HA\otimes \HAb$ and a subalgebra $M$ acting on it, but it gives no assurance that any of them actually holds.  From the point of view of quantum error correction, subregion duality (meaning the existence of $O_A$ and $O'_{\Ab}$) is guaranteed for a subalgebra code with complementary recovery on $A$ and $\Ab$, and the RT formula and equivalence of relative entropies then follow.  In holography however, we do not yet have an explicit bulk algorithm for subregion duality when the entanglement wedge is larger than the causal wedge.  So we must instead rely on the derivations of \cite{Lewkowycz:2013nqa,Faulkner:2013ana} to establish the RT formula, after which we may use theorem \ref{bigrthm} to establish subregion duality \cite{Dong:2016eik}.  It would be nice to have a direct understanding of subregion duality from the bulk point of view, not requiring a detour through the RT formula.   

Secondly, although theorems \ref{bigrthm}, \ref{othm} give a rather general characterization of subalgebra correctability with complementary recovery for a fixed factorization $\HA\otimes \HAb$, something that would really be nice is a condition on $\Hc$ which guarantees subalgebra correctability with complementary recovery for \textit{arbitrary} regions $A$ and $\Ab$.  This is clearly a much stronger constraint on $\Hc$ than correctability for a particular $A$, but we expect it to hold in AdS/CFT.  We have seen that this requires substantial entanglement in the $|\chi_\alpha\ran$'s, but that is far from giving a necessary and sufficient condition for which subspaces have this property.  

Thirdly, even once we have established subalgebra correctability with complementary recovery, and thus the existence of an operator $\LA$ for which the RT formula holds, in general we do not expect $\LA$ to have an interpretation as extremizing something (such as the area).  This must be a special property of holographic codes, and it would be interesting if a more general condition could be given under which $\LA$ has an extremal (or minimal) interpretation.

Finally, theorem \ref{bigrthm} as stated only applies to holography in detail to order $G^0$.  This already tells us that we really need an approximate version of theorem \ref{bigrthm}, but actually the situation gets worse at higher orders in gravitational perturbation theory.  The reason is that the RT formula itself is modified, and my results need to be refined to account for this.  We do not yet know in detail how to modify it, but one proposal has been given in \cite{Engelhardt:2014gca}.  The idea is that we locate the surface $\gamma_A$ by extremizing right hand side of the RT formula, being careful to include the higher order corrections to $\LA$.  This has the effect of making $\LA$ a nonlinear operator, which makes it difficult to define the algebra $M$ in a way that $M'$ corresponds to the operators in the complementary entanglement wedge (it is no longer possible to do a gauge-fixing that puts $\gamma_A$ at a definite coordinate submanifold such as the one described in \cite{Jafferis:2015del}).  We can define a subalgebra $M$ by requiring that its elements are in $\EA$ for any state in $\Hc$, but then $M'$ will include some operators that are not strictly supported in $\mathcal{E}_{\Ab}$.  There will be a ``no-man's land'' of Planckian size consisting of operators which are sometimes in $\EA$ and sometimes in $\mathcal{E}_{\Ab}$, and it will in general get mixed up with the center of $M$ in defining $\LA$.  I don't see any fundamental problem with some version of \ref{bigrthm} holding at higher orders in $G$, but it will clearly need to take these issues into account.

\paragraph{Acknowledgments} I would like to thank Ahmed Almheiri, Ning Bao, Tom Banks, Cedric Beny, Horacio Casini, Thomas Dumitrescu, Xi Dong, Daniel Jafferis, Matt Headrick, Aitor Lewkowycz, Juan Maldacena, Don Marolf, Greg Moore, Hirosi Ooguri, Jonathan Oppenheim, Lenny Susskind, Andy Strominger, Aron Wall, Beni Yoshida, and Sasha Zhiboedov for very useful discussions. I'd also like to thank the Yukawa Institute for Theoretical Physics at Kyoto University and the University of Amsterdam for hospitality while this work was being completed.  I am supported by DOE grant  DE-FG0291ER-40654 and the Harvard Center for the Fundamental Laws of Nature.
  
\appendix
\section{Von Neumann algebras on finite-dimensional Hilbert spaces}\label{vnapp}
Von Neumann algebras are a beautiful subject, but unfortunately most discussions in the mathematics literature are greatly complicated by an insistence on treating the infinite-dimensional case from the beginning.  In physics, it is usually true that the finite-dimensional case is enough for any practical applications: even in cases where the Hilbert space is infinite-dimensional, such as in quantum field theories, there is almost always a way of truncating the theory to a finite-dimensional Hilbert space without losing any important data for the problem being considered.  In this appendix I will present the theory of von Neumann algebras on finite-dimensional Hilbert spaces, with the goal being to save the reader the trouble of extracting these results from the infinite-dimensional literature.  The source from which I found this extraction the easiest is \cite{Jones}, whose presentation I have followed fairly closely.  The reader is encouraged to look there for many more results, and may also wish to consult \cite{beny2015algebraic} for a recent explanation of some of these results from a $C^*$-algebra point of view. 
\subsection{Definitions}
Say that $\Hh$ is a finite-dimensional Hilbert space, ie a finite-dimensional complex vector space with an inner product, and $\LH$ is the set of linear operators acting on $\Hh$.  I'll denote the identity operator in $\Hh$ as $I$.
\begin{mydef}\label{vndef}
A \textbf{von Neumann algebra on $\bf{\Hh}$} is a set $M\subseteq \Ll(\Hh)$ such that:
\bi
\item $\forall \lambda \in \mathbb{C}$, $\lambda I\in M$ 
\item $\forall x \in M$, $x^\dagger \in M$ 
\item $\forall x,y\in M$, $xy\in M$ 
\item $\forall x,y \in M$, $x+y\in M$.
\ei
\end{mydef}
In other words it is a set of linear operators on $\Hh$ which is closed under hermitian conjugation, addition, multiplication, and which contains all scalar multiples of the identity operator.  Any von Neumann algebra automatically induces two other natural von Neumann algebras on $\mathcal{H}$:
\begin{mydef}
Given a von Neumann algebra $M$ on $\mathcal{H}$, the \textbf{commutant of M} is defined as $M'\equiv \{y\in \LH | xy=yx, \forall x\in M\}$.
\end{mydef}
\begin{mydef}
Given a von Neumann algebra $M$ on $\mathcal{H}$, the \textbf{center of M} is defined as $Z_M\equiv M\cap M'$.
\end{mydef}
In other words the commutant is the set of all linear operators that commute with everything in $M$, and the center is the subset of those which are themselves in $M$.  It is straightforward to confirm that they are in fact von Neumann algebras on $\Hh$.

\subsection{Projections and partial isometries}
In our study of von Neumann algebras, it will be very convenient to introduce the notions of \textit{projection} and \textit{partial isometry}:
\begin{mydef}
A linear map $p\in \LH$ is called a \textbf{projection} if $p^\dagger=p$ and $p^2=p$.
\end{mydef}
\begin{mydef}
A linear map $a\in \LH$ is called a \textbf{partial isometry} if $a^\dagger a=p$, where $p$ is a projection.  
\end{mydef}
A projection always has a subspace $p\mathcal{H}$ on which it acts identically, and whose orthogonal complement $(1-p)\Hh$ it annihilates.  Partial isometries are characterized by following theorem:
\begin{thm}\label{pithm}
Say that $a$ is a partial isometry on $\Hh$, obeying $a^\dagger a=p$, with $p$ a projection. Then $a^\dagger$ is also a partial isometry, obeying $a a^\dagger =q$, with $q$ also a projection, and there exists a unitary operator $u\in \LH$ such that $q=upu^\dagger$.  Thus $q$ and $p$ have equal rank, and in fact we can choose $u$ so that $a=up$.
\end{thm}
\begin{proof}
To see that $a^\dagger$ is a partial isometry, we first observe that any $|v\ran\in (1-p)\Hh$ is also annihilated by $a$, since $\lan v|a^\dagger a|v\ran=0$.  If we represent $a$ in block form using the direct sum decomposition $\mathcal{H}=p\Hh \oplus (1-p)\Hh$, only the first column can thus be nonzero: $a=\begin{pmatrix} A && 0 \\ B && 0\end{pmatrix}$, with $A^\dagger A+B^\dagger B=I_{p\Hh}$.  Using this expression for $a$, it is easily confirmed that $(aa^\dagger)^2=aa^\dagger$, and thus that $q$ is a projection.  

To see that $q$ and $p$ have equal rank, we can first observe that for any $|v\ran \in p\Hh$, we have $qa|v\ran=a|v\ran$, and thus $a|v\ran\in q\Hh$.  We also have that $\lan v_1|a^\dagger a|v_2\ran=\lan v_1|v_2\ran$ for all $|v_1\ran,|v_2\ran\in p\Hh$, so applying $a$ to an orthonormal basis for $p\Hh$ we see that we must have $\mathrm{dim}(p\Hh)\leq \mathrm{dim}(q\Hh)$.  Making the same argument acting on elements of $q\Hh$ with $a^\dagger$, we then conclude that $\mathrm{dim}(p\Hh)\geq \mathrm{dim}(q\Hh)$, and thus that $p$ and $q$ have equal rank.  Any two projections of equal rank are always unitarily equivalent, so indeed we have $q=upu^\dagger$ for some $u$ a unitary in $\LH$.  Moreover we can choose $u$ so that $u|v\ran=a|v\ran$ for any $|v\ran\in p\Hh$, in which case we have $a=up$.
\end{proof}
Another important property of partial isometries is their role in the polar decomposition theorem:
\begin{thm}\label{pdthm}
Say that $x\in \LH$.  Then we have $x=a|x|$, where $|x|$ is a non-negative matrix and $a$ is a partial isometry such that $a^\dagger a\equiv p$ is the projection onto the orthogonal complement of the kernel of $x$. Moreover both $a$ and $|x|$ are unique. 
\end{thm}
\begin{proof}
We first define $|x|\equiv \sqrt{x^\dagger x}$, which is clearly non-negative.  It has the same kernel as $x$, since $\lan v|x^\dagger x|v\ran=0 \Leftrightarrow x|v\ran=0$.  Now $|x|$ is invertible on $p\Hh=ker(x)^\perp$, so defining $a\equiv x\left(|x|^{-1}\oplus 0_{ker(x)}\right)$, we see that $a^\dagger a=\left(|x|^{-1}\oplus 0_{ker(x)}\right)|x|^2\left(|x|^{-1}\oplus 0_{ker(x)}\right)=p$.  $|x|$ is clearly unique, since if $x=a|x|$ then $x^\dagger x=|x|^2$.  $a$ is also unique, since if $a'|x|=a|x|$ we can multiply on both sides on the right by $\left(|x|^{-1}\oplus 0_{ker(x)}\right)$ to conclude that $a=a'$.
\end{proof}
Note that the restriction on the kernel of $a$ in this theorem is crucial for its uniqueness.  We could instead ask for $a$ to be unitary, and in fact the polar decomposition theorem is often stated that way, but then $a$ would not be unique.  
\subsection{The bicommutant theorem}
Perhaps the most fundamental theorem about von Neumann algebras is von Neumann's famous bicommutant theorem:
\begin{thm}
For any von Neumann algebra $M$ on $\Hh$, we have $M'' \equiv (M')'=M$.
\end{thm}
\begin{proof}
The proof of this theorem is based on a clever ``doubling'' trick: rather than considering the action of $M$ directly on $\Hh$, we instead extend it to a new von Neumann algebra $I\otimes M$ on $\Hh\otimes \Hh$.  If we denote $n=\mathrm{dim}(\Hh)$, then we can view elements of $\Ll(\Hh\otimes \Hh)$ as $n\times n$ block matrices, whose blocks are themselves $n\times n$ matrices.  Elements of $I\otimes M$ are block diagonal in this representation, with the same element $x\in M$ in each diagonal block.  In other words we are doing a block decomposition based on the fact that $\Hh \otimes \Hh \cong \oplus_{i=1}^n \Hh_i$.  In this extended Hilbert space, it is easy to confirm that $(I\otimes M)'$ is the set of $n\times n$ block matrices whose blocks are arbitrary elements of $M'$.  By considering particular elements of $(I\otimes M)'$ where all blocks are zero except for one, which is taken to be $I$, we can also see that $(I\otimes M)''$ is the set of block diagonal matrices with the same element $z \in M''$ in each diagonal block.

Now consider an arbitrary vector $|v\ran\in \Hh\otimes\Hh$. We can define a subspace $V\in \Hh\otimes \Hh$ via $V\equiv(I\otimes M)|v\ran$, ie $V$ is the set of all vectors we can reach by acting on $|v\ran$ with an element of $I\otimes M$.  The key point is to observe that the projection $p_V$ onto $V$ commutes with all elements of $I\otimes M$, and is thus an element of $(I\otimes M)'$.  This is true because $I\otimes M$ acts within $V$, $I\otimes M$ is spanned by its hermitian elements, and any hermitian operator that preserves a subspace must commute with the projection onto that subspace.  This then implies that $p_V$ commutes with everything in $(I\otimes M)''$, which implies that any element $\begin{pmatrix} z &&0 &&\cdots\\ 0 &&z && \cdots \\ \vdots &&\vdots&&\ddots  \end{pmatrix}$ of $(I\otimes M)''$ must preserve $V$, and in particular acting on $|v\ran$ must be equivalent the action of some element $\begin{pmatrix} x &&0 &&\cdots\\ 0 &&x && \cdots \\ \vdots &&\vdots&&\ddots  \end{pmatrix}$ of $I\otimes M$.  But if we choose $|v\ran=\oplus_i |v_i\ran$ for some basis $|v_i\ran$ of $\Hh$, this then implies that $z=x$, and thus that $M''\subseteq M$.  Since $M''\supseteq M$ by definition, this establishes $M''=M$.
\end{proof}

It is interesting to note that the inclusion of scalar multiples of the identity operator in $M$ is essential for this proof: otherwise the vector $|v\ran$ might not be in the subspace $V$, so we would not be able to conclude that $z|v\ran\in V$.  If we replace the first condition in def. \eqref{vndef} by the weaker condition that for any $x\in M$ and $\lambda \in \mathbb{C}$, we have $\lambda x \in M$, the object we define instead is a representation of a $C^*$-algebra on $\mathcal{H}$.\footnote{$C^*$-algebras are defined abstractly, so we need to specify that we are representing one as a subalgebra of $\LH$.  Also note that in infinite dimensions there is an additional distinction between von Neumann algebras and representations of $C^*$-algebras on $\mathcal{H}$: one requires the algebra to be closed under different topologies in the two cases.} Indeed representations of $C^*$-algebras do not in general obey the bicommutant theorem!  A simple counter-example is the set of scalar multiples of some projection $p$ of non-maximal rank on $\Hh$: the identity $I$ is not in this representation of the abstract $C^*$-algebra isomorphic to $\mathbb{C}$, but it is in its bicommutant.\footnote{A linguistic subtlety here is that the abstract $C^*$-algebra $\mathbb{C}$ is actually unital, in the sense of containing an element that acts identically on all other elements, but this element is represented on $\Hh$ as the nonmaximal projection $p$.}  In physical applications we usually think of subalgebras as reflecting restrictions on what an observer can measure: since the identity corresponds to not measuring anything, it should be accessible to any observer, and thus we should always include it.

\subsection{Basic properties of von Neumann algebras}
Let's now establish some more basic facts about von Neumann algebras:
\begin{prop}\label{projprop}
Say that $x\in M$ is hermitian.  Then the projections onto the eigenspaces of $x$ are also elements of $M$.  Moreover if $f$ is a function $f:D \rightarrow\mathbb{C}$, with $D\subseteq \mathbb{R}$, and all eigenvalues of $x$ are in $D$, then the operator $f(x)$ is also in $M$.
\end{prop}
\begin{proof}
Each eigenspace projection of $x$ must commute with any $y\in M'$, since otherwise $y$ would not commute with $x$, but this means that the projections are in $M''$, which by the bicommutant theorem is equal to $M$.  Once we have the projections, we can define $f(x)$ by applying $f$ to each eigenvalue in the spectral representation of $x$, and since this is a sum over elements of $M$ times elements of $\mathbb{C}$, it must also be in $M$.
\end{proof}
\begin{prop}\label{uprop}
Any element of $M$ can be written as a linear combination of four unitary elements of $M$.
\end{prop}
\begin{proof}
We've already observed that any $x\in M$ can be written as a linear combination of two hermitian operators, explicitly we have $x=\frac{x+x^\dagger}{2}+i\frac{x-x^\dagger}{2i}$.  So it is enough to consider the case where $x^\dagger=x$.  We can rescale $x$ so that its largest eigenvalue has absolute value less than one, in which case we have $x=\frac{1}{2}\left(x+i\sqrt{1-x^2}\right)+\frac{1}{2}\left(x-i\sqrt{1-x^2}\right)$.  The operators $x\pm i \sqrt{1-x^2}$ are clearly unitary, and by proposition \eqref{projprop} they are elements of $M$.
\end{proof}
\begin{prop}
Say that $p$ is a projection in $M$.  Then $pMp$ defines a von Neumann algebra on $p\Hh$, and its commutant on $p\Hh$ is $M'p$.
\end{prop}
\begin{proof}
It is straightforward to confirm that $pMp$ is a von Neumann algebra, for example $(px_1p)(px_2p)=p(x_1 p x_2)p$.  To find the commmutant, first note that by the bicommutant theorem it is enough to show that $pMp=(M'p)'$.  Indeed say that $x$ on $p\Hh$ commutes with $yp$ for all $y\in M'$.  If we define $\hat{x}\equiv x\oplus 0_{(1-p)\Hh}$, then clearly $x=p\hat{x}p$.  We now need to show that $\hat{x}\in M$.  Again using the bicommutant theorem, we just need to see that $\hat{x}$ commutes with any $y\in M'$, since it will then be in $M''=M$.  But notice that $\hat{x}y=\hat{x}py=\hat{x}yp=yp\hat{x}=y\hat{x}$, so we are done.  
\end{proof}
\begin{prop}
Say that $x\in M$, and that $x=a|x|$ is the unique polar decomposition of $x$ promised by theorem \eqref{pdthm}.  Then $a$ and $|x|$ are both also in $M$.
\end{prop}
\begin{proof}
$|x|\equiv \sqrt{x^\dagger x}$ is clearly in $M$ by proposition \eqref{projprop}.  We will show that $a$ is also in $M$ by showing that it commutes with everything in $M'$, and then again resorting to the bicommutant theorem.  In fact by proposition \eqref{uprop}, it is sufficient to show that it commutes with any unitary element $u$ of $M'$.  We can first note that $ua|x|=a|x|u=au|x|$, since $x$ and $|x|$ are both in $M$.  But since $|x|$ is in $M$, by proposition \eqref{uprop} the projection $a^\dagger a$ onto the orthogonal complement of its kernel must also be in $M$.  This means that $(au)^\dagger(au)=u^\dagger a^\dagger a u=a^\dagger a=(ua)^\dagger (ua)$, so by the uniqueness of the polar decomposition of $ua|x|$ we must have $ua=au$.
\end{proof}

\subsection{Factors}
In the theory of von Neumann algebras, there is a special role for algebras with trivial center:
\begin{mydef}
A von Neumann algebra $M$ on $\Hh$ is called a \textbf{factor} if its center $Z_M\equiv M\cap M'$ contains only scalar multiples of $I$.
\end{mydef}
In the following section we will understand the origin of this name.  Factors have several nice properties:
\begin{prop}
Say that $M$ is a factor, and that $p$ and $q$ are nonzero projections in $M$.  Then there exists a unitary operator $u\in M$ such that $puq\neq 0$.  
\end{prop}
\begin{proof}
Say that $puq=0$ for all unitaries $u$ in $M$.  Then we would also have $u^\dagger p uq=0$.  But now say we define a new projection operator $r$ by the property that it annihilates only those elements of $\Hh$ which are in the kernel of $u^\dagger p u$ for all $u\in M$.  $r$ is not the identity, since any vector in $q\Hh$ must be annihilated by all $u^\dagger p u$ and thus by $r$, and $r$ is nonzero since $u^\dagger p u$ is nonzero.   The kernel of $r$ is apparently preserved by the action of any $u$.  But this means $r$ commutes with all $u$, and thus with everything in $M$ by proposition \eqref{uprop}.  $r$ is also in $M$, since it commutes with everything in $M'$ (otherwise there would be some $u^\dagger p u$ whose kernel was not preserved by a hermitian element of $M'$, which would contradict $u^\dagger p u\in M$.)  But these things together contradict the assumption that $M$ is a factor, since we have shown that $r$ is a nontrivial element of the center. 
\end{proof}
Before stating the next property, it is convenient to introduce an ordering notation on projections.  Say that $p$ and $q$ are projections. If $p\Hh\subseteq q \Hh$, or equivalently $ker(p)\supseteq ker(q)$, then we say $p\leq q$.
\begin{prop}\label{aprop}
Say that $M$ is a factor, and that $p$ and $q$ are nonzero projections in $M$.  Then there exists a partial isometry $a$ such that $a^\dagger a\leq q$ and $aa^\dagger \leq p$.
\end{prop} 
\begin{proof}
Define $x\equiv puq$, with $u\in M$ chosen so that $x\neq 0$.  By the polar decomposition theorem, we have $x=a|x|$, with $ker(a)=ker(|x|)$.  Clearly if $|v\ran$ is annihilated by $q$ it is annihilated by $a|x|$, and thus by $a$, so we have $a^\dagger a \leq q$.  Moreover since $qu^\dagger p=|x|a^\dagger$, we see that if $|v\ran$ is annihilated by $p$ it must also be annihilated by $|x|a^\dagger$.  From theorem \eqref{pithm}, we know that $a^\dagger=a^\dagger a w^\dagger$, with $w$ a unitary that maps the kernel of $|x|$ to that of $a^\dagger$, so $|x|a^\dagger|v\ran=0\implies a^\dagger|v\ran=0$.  Therefore we have $aa^\dagger \leq p$.
\end{proof}
In studying factors, it is convenient to introduce a special kind of projection:
\begin{mydef}
Say $M$ is a von Neumann algebra on $\Hh$, and $p$ is a nonzero projection.  We say that $p$ is a \textbf{minimal projection} if for any projection $q\in M$, we have $q\leq p$ if and only if $q=0$ or $q=p$.  
\end{mydef}
Since $\Hh$ is finite-dimensional, minimal projections must always exist in any von Neumann algebra. Indeed given any nonzero nonminimal projection $p$, we can find a nonzero projection $q$ of smaller rank such that $q\leq p$.  If $q$ is nonminimal then we can do this again, and since any projection of rank one is minimal, this procedure always eventually finds a minimal projection.  We can characterize minimal projections by the following theorem:  
\begin{thm}\label{minprojthm}
Say that $M$ is a von Neumann algebra on $\Hh$, and $p$ is a minimal projection.  Then $pMp=\mathbb{C}p$, or in other words $pMp$ consists only of scalar multiples of $p$.  
\end{thm}
\begin{proof}
$pMp$ will always contain $\mathbb{C}p$.  If it contains any other operators, then by proposition \eqref{projprop} it will have a nontrivial projection $q$.  But such a $q$ would contradict the minimality of $p$.
\end{proof}
The existence of minimal projections is a key point where our insistence that $\Hh$ be finite-dimensional is essential.  In the infinite-dimensional case, factors that contain a minimal projection are called factors of type I, while those that don't are called factors of types II and III.  Perhaps the main thing we achieve by considering only the finite-dimensional case is that we do not need to consider these more complicated factors.\footnote{The difference between type II and type III is based on the existence of \textit{finite projections}: a projection $p\in M$ is called finite if there is no other projection $q\in M$ obeying $q<p$, but nonetheless having a partial isometry $a$ such that $p=a^\dagger a$ and $q=aa^\dagger$.  Any projection for which $p\Hh$ is finite-dimensional is always finite, but there can in general be finite projections with $p\Hh$ still infinite-dimensional.  A von Neumann algebra is called type II if it has no minimal projections but does have a finite projection, and type III if it has neither.}
 
\subsection{The classification of von Neumann algebras on finite-dimensional Hilbert spaces}
We are now in a position to classify all von Neumann algebras on finite-dimensional Hilbert spaces.  The most challenging step turns out to be the classification of factors, so we will discuss this first.  
\begin{thm}\label{factorthm}
Say that $M$ is a factor on $\Hh$.  Then there exists a tensor factorization $\Hh=\HA\otimes\HAb$ such that $M=\Ll(\HA)\otimes I_{\ol{A}}$.  In other words, $M$ is just the set of all linear operators on some tensor factor $\HA$ of $\Hh$.  Moreover we have $M'=I_A\otimes \Ll(\HAb)$.  
\end{thm}
\begin{proof}
The basic idea is to consider a maximal set of minimal projections $p_i$ such that $p_ip_j=0\,\, \forall i\neq j$.  Such a set always exists, since we can take any single minimal projection and then keep including more until we no longer can.  The first thing to show is that there is no state which is annihilated by all the $p_i$.  If there were, then we could define a nonmaximal projection $r$ which annihilates only those states annihilated by all the $p_i$.  This $r$ would be in $M$, since it must commute with everything in $M'$, and it would obey $(1-r)p_i=0 \,\,\forall i$.  But this would contradict the maximality of the $p_i$, since we could then include $(1-r)$ into the set.  Thus we must have $I=\sum_i p_i$.  

Now by proposition \eqref{aprop}, for any $i$ we must have a nonzero partial isometry $a_i$ such that $a_i^\dagger a_i \leq p_i$ and $a_i a_i^\dagger \leq p_1$.  By the minimality of $p_1$ and $p_i$, we must in fact have that $a_i^\dagger a_i=p_i$ and $a_i a^\dagger_i=p_1$.  By theorem \eqref{pithm}, we see that the $p_i$ are all unitarily equivalent, and thus have equal rank.  Moreover since $I=\sum_i p_i$, this rank must divide the dimensionality of $\Hh$.  We will soon see that in fact the $p_i$ are the projections onto an orthonormal basis of a factor $\HA$ tensored with the identity on $\HAb$.      

Indeed we can now observe that since $I=\sum_i p_i$, for any $x\in M$ we have $x=\sum_{ij}p_ixp_j$.  Moreover since $a_i$ maps $p_i\Hh$ to $p_1\Hh$, we have $p_i x p_j=a_i^\dagger a_i x a_j^\dagger a_j=a_i^\dagger p_1a_i x a_j^\dagger p_1 a_j$.   Since $p_1$ is minimal, by theorem \eqref{minprojthm} we have $p_1a_i x a_j^\dagger p_1=\lambda_{ij} p_1$ for some coefficients $\lambda_{ij}\in \mathbb{C}$, and thus $p_i x p_j=\lambda_{ij} a_i^\dagger p_1 a_j=\lambda_{ij} a_i^\dagger a_j$.  We then have $x=\sum_{ij}\lambda_{ij}a^\dagger_i a_j$, so the $a_i$'s apparently generate all of $M$.

Finally we need to identify the algebra generated by the $a_i$'s.  Let's first notice that if we block decompose $\Hh=\oplus_i p_i \Hh$, then we can define a tensor product stucture $\Hh=\HA\otimes \HAb$ by taking $\Ll(\HA)\otimes I_{\ol{A}}$ to be the set of block matrices where each block is an arbitrary multiple of the identity on that block, and taking $I_A\otimes \Ll(\HAb)$ be the set of block diagonal matrices with the same element of $\Ll(\HAb)$ in each diagonal block.  We can choose a basis within each block so that $a_i$ is represented as a matrix with an identity operator in the $i$th column of the first row and zeros elsewhere, in which case the operator $a_i^\dagger a_j$ will have the identity in the $j$th column of the $i$th row and be zero otherwise.  But these matrices clearly generate all of $\Ll(\HA)\otimes I_{\ol{A}}$, which is thus equal to $M$.  Moreover by studying the commutator of an arbitrary matrix with $a^\dagger_i a_j$, it is straightforward to confirm that $M'=I_A\otimes \Ll(\HAb)$.
\end{proof}
This theorem clearly justifies the definition of a factor, although in infinite dimensions there are factors (of types II and III) for which it isn't true.  Since the proof was somewhat involved, I'll quickly recap the logic.  By considering projections that are in $M$, we study subspaces which $M$ ``knows about''.  Any two subspaces of equal dimensionality in $\Hh$ are isomorphic, but $M$ only ``knows about'' this if the partial isometry that relates them is in $M$.  Any factor has the property that its minimal projections are all related by partial isometries in $M$, which is a kind of irreducibility of $M$.  Moreover in a factor there is a maximal set of minimal projections which are mutually orthogonal and complete.  We can use this set to factorize $\mathcal{H}=\HA\otimes \HAb$, and use the isometries between the projections to generate $\Ll(\HA)$.  

Now we consider the general case, where $M$ is not necessarily a factor.  The basic point however is that since all elements of the center $Z_M$ are mutually commuting, we can simultaneously diagonalize them.  By proposition \eqref{projprop} this means there is a family of projections $p_\alpha \in Z_M$, obeying $p_\alpha p_\beta=0$ for any $\alpha\neq \beta$, such that $Z_M$ is equivalent to the set of operators $\sum_\alpha \lambda_{\alpha} p_\alpha$, with $\lambda_\alpha$ an arbitrary set of complex numbers.  We then have the following proposition:
\begin{prop}
Say that $M$ is a von Neumann algebra, whose center $Z_M$ is spanned by the projections $p_\alpha$, obeying $p_\alpha p_\beta=0$ for all $\alpha\neq \beta$.  Then for all $\alpha$, $p_\alpha M p_\alpha$ is a factor on $p_\alpha\Hh$.  Moreover if $\alpha\neq\beta$ then $p_\alpha M p_\beta=0$.
\end{prop}
\begin{proof}
Say that $p_\alpha M p_\alpha$ had a nontrivial central element $c$.  Then $c\oplus 0_{(1-p_\alpha) \Hh}$ would be an element of $Z_M$ that was not in the span of the $p_\alpha$'s, but we have chosen them to span $Z_M$ so no such $c$ can exist.  Thus $p_\alpha M p_\alpha$ is a factor.  Moreover if $\alpha\neq\beta$, then $p_\alpha M p_\beta=Mp_\alpha p_\beta=0$.  
\end{proof}
This proposition says that if we decompose $\Hh=\oplus_\alpha p_\alpha \Hh$, then every element of $M$ is block diagonal, and moreover each diagonal block is a factor algebra.  Together with theorem \eqref{factorthm}, this at last implies the classification theorem:
\begin{thm}\label{classthm}
Say that $M$ is a von Neumann algebra on $\Hh$, with $dim(\Hh)<\infty$. Then we have a block decomposition $\Hh=\oplus_\alpha \left(\mathcal{H}_{A_\alpha}\otimes \mathcal{H}_{\ol{A}_\alpha}\right)$, in terms of which $M$ and $M'$ are block-diagonal, with decompositions $M=\oplus_\alpha\left(\Ll\left(\mathcal{H}_{A_\alpha}\right)\otimes I_{\ol{A}_\alpha}\right)$ and $M'=\oplus_\alpha \left(I_{A_\alpha}\otimes\Ll\left(\mathcal{H}_{\ol{A}_\alpha}\right)\right)$.
\end{thm}
In stating this theorem I have introduced a convenient abuse of notation, whereby if we have a block diagonal operator with diagonal blocks $x_{\alpha}$, then we can write $x=\oplus_\alpha x_\alpha$.  In infinite dimensions this theorem has a partial analogue: any von Neumann algebra is a ``direct integral'' of factor algebras.  The classification of factors however is much more complicated, with type III factors being the most difficult.  In fact the type III case is what one expects for the algebra of operators in a finite region in a continuum quantum field theory \cite{araki1964type,driessler1977type,haag2012local}. This problem can be avoided by working in a cutoff theory: this includes many additional states whose continuum limits would have had infinite energy, including those necessary to return the algebra to type I.    

\subsection{Entropy}\label{entapp}
I'll now discuss the notion of the entropy of a state on a von Neumann algebra \cite{ohya2004quantum,Casini:2013rba}.  Several new ideas are needed, so we'll proceed in stages.
\subsubsection{States}
So far we have discussed von Neumann algebras as subsets of the linear operators $\Ll(\Hh)$ on a (finite-dimensional) Hilbert space $\Hh$.  In quantum mechanics hermitian elements of $\LH$ correspond to observables, but to do physics we also need to introduce the notion of states:
\begin{mydef}
A linear operator $\rho\in \LH$ is called a \textbf{state on $\LH$} if it is hermitian, non-negative, and has $\Tr\rho=1$.  
\end{mydef}
Any state $\rho$ has a natural linear action $\mathbb{E}_\rho$ on  $\LH$.\footnote{In fact in the mathematical literature states are usually defined as linear, non-negative maps on $\LH$, obeying $\mathbb{E}_\rho(I)=1$.  From the point of view of this article, this is needlessly abstract.}
 For any $x\in \LH$, we define
\be
\mathbb{E}_\rho (x)=\Tr(\rho x).
\ee
In quantum mechanics, if $x$ is hermitian then $\mathbb{E}_\rho(x)$ is called the \textit{expectation value of the operator $x$ in the state $\rho$}.
It is often the case that one is interested only in observables that are elements of some von Neumann algebra $M$.  A generic state $\rho$ will not necessarily be an element of $M$, and will typically contain more information than is needed to compute expectation values of elements of $M$.  The following theorem gives an elegant way to discard this extra information:
\begin{thm}\label{rhoM}
Say that $M$ is a von Neumann algebra on $\Hh$, and $\rho$ is a state on $\Hh$.  Then there exists a unique state $\rho_M\in M$ such that $\mathbb{E}_\rho(x)=\mathbb{E}_{\rho_M}(x)$ for all $x\in M$.
\end{thm}
\begin{proof}
The basic idea is to define 
\be\label{uint}
\rho_M\equiv \int_{u\in M'} du\, u \rho u^\dagger.
\ee
Here we are integrating over the set of unitary elements of $M'$, using the invariant Haar measure $du$ on this compact group.\footnote{To see that the unitary subgroup of $M'$ is compact, note that any Cauchy-convergent sequence of unitary elements $u_n \in M'$ will converge to some unitary $u$ by the compactness of the unitary group, and by continuity of the commutator the limit will also be in $M'$.  This argument is straightforward in finite dimensions, it would be more complicated otherwise.}  $\rho_M$ is clearly hermitian, non-negative, and has trace one.  To show that it is an element of $M$, we will argue that it commutes with any unitary $v$ in $M'$, and thus is in $M''=M$ by proposition \eqref{uprop} and the bicommutant theorem.  Indeed say that $v\in M'$ is unitary.  Then we have
\be
v\rho_M=\int_{u\in M'} du \,vu\rho u^\dagger=\int_{u'\in M'} du' \,u'\rho u'^\dagger v=\rho_M v,
\ee 
where in the middle we have changed variables $u'=vu$ and used the invariance of the measure.  Finally to see that $\rho_M$ is unique, say that there existed $\rho_M'\neq \rho_M$ also obeying the results of the theorem.  Then we must have $\Tr\left((\rho_M-\rho_M')x\right)=0$ for all $x\in M$.  But in particular we can take $x=\rho_M-\rho_M'$, which then tells us that $\Tr (\rho_M-\rho_M')^2=0$, and thus that $\rho_M=\rho_M'$.  
\end{proof}
This theorem says that for the purpose of computing expectation values of $M$, we can always replace any state by an element of $M$.  To develop some intuition, let's compute $\rho_M$ for the case where $M$ is a factor. By theorem \eqref{factorthm} we know that there is a factorization $\Hh=\HA\otimes \HAb$ such that $M=\Ll(\HA)\otimes I_{\Ab}$.  If we define the reduced state
\be
\rho_A\equiv \Tr_{\Ab}\rho,
\ee
then it is easy to see that the operator
\be\label{rhoM2}
\rho_M\equiv \rho_A\otimes \frac{I_{\Ab}}{|\Ab|}
\ee
obeys the results of theorem \eqref{rhoM}.  By the uniqueness result of that theorem, the $\rho_M$ defined by eq. \eqref{rhoM2} must be equivalent to the one defined by eq. \eqref{uint}.\footnote{This equivalence isn't hard to show explicitly, using standard unitary integration technology which, for example, is reviewed in appendix D of \cite{Harlow:2014yka}.}

For a general von Neumann algebra $M$ we can also write down an explicit representation for $\rho_M$ along similar lines.  From theorem \eqref{classthm}, we know that there is a decomposition 
\be\label{hilbertapp}
\Hh=\oplus_\alpha \left(\mathcal{H}_{A_\alpha}\otimes \mathcal{H}_{\ol{A}_\alpha}\right),
\ee
in terms of which we have
\be\label{opapp}
M=\oplus_\alpha\left(\Ll\left(\mathcal{H}_{A_\alpha}\right)\otimes I_{\ol{A}_\alpha}\right).
\ee
Any state $\rho$ can be written in block form with respect to the direct sum in eq. \eqref{hilbertapp}, and only blocks which are diagonal in $\alpha$ will contribute to expectation values of elements of $M$.  From each diagonal block, we can define
\be
p_\alpha\rho_{A_\alpha}\equiv \Tr_{\Ab_{\alpha}}\rho_{\alpha \alpha}.
\ee
Here $p_\alpha$ is a positive number chosen so that $\Tr_{A_\alpha} \rho_{A_\alpha}=1$.  The condition $\Tr \rho=1$ implies that $\sum_\alpha p_\alpha=1$.  Finally we can then define the block-diagonal state
\be\label{rhoM3}
\rho_M\equiv \oplus_\alpha \left(p_\alpha\rho_{A_\alpha}\otimes \frac{I_{\Ab_\alpha}}{|\Ab_\alpha|}\right),
\ee
which again is easily shown to obey the results of theorem \eqref{rhoM}: it is hermitian, non-negative, has trace one, is of the form \eqref{opapp} and is thus in $M$, and clearly gives the same expectation values as $\rho$ for elements of $M$.

\subsubsection{Modified trace and entropy}
We see from eq. \eqref{rhoM2} that, when $M$ is a factor, the state $\rho_M$ is closely related to the reduced state $\rho_A$.  The reduced state $\rho_A$ is what is usually used to define the von Neumann entropy $S(\rho_A)\equiv -\Tr_A \rho_A \log \rho_A$.  This suggests a natural generalization to the entropy of a state $\rho$ on an arbitrary von Neumann algebra $M$:
\be\label{entropy}
S(\rho,M)\equiv -\sum_\alpha \Tr_{A_\alpha} \left(p_\alpha \rho_{A_\alpha}\log (p_\alpha \rho_{A_\alpha})\right)=-\sum_\alpha p_\alpha \log p_\alpha +\sum_{\alpha}p_\alpha S(\rho_{A_\alpha}).
\ee
For practical purposes we could simply take \eqref{entropy} as the definition of the entropy, and check that it has the properties we expect an entropy to have.  It would be preferable however to arrive at this expression from a more abstract point of view, and in particular it would be nice to avoid making explicit use of the decomposition \eqref{hilbertapp}. Readers who are already satisfied with \eqref{entropy} may skip to the next subsection for a discussion of the properties of this definition.\footnote{Beni Yoshida has pointed out to me that the state $\rho_M$ and the entropy \eqref{entropy} arise rather naturally in attempts to defined coarse-grained entropy \cite{GellMann:2006uj}, and that it would be interesting to understand this better in the context of \cite{Kelly:2013aja}.}

What would be ideal is to extract this entropy from the state $\rho_M$, but things are not as simple as computing $-\Tr \rho_M \log \rho_M$: already when $M$ is a factor, from \eqref{rhoM2} this apparently disagrees with the standard entropy by $\log |\Ab|$.  If $M$ is not a factor, then from \eqref{rhoM3} the disagreement with \eqref{entropy} is apparently $\sum_\alpha p_\alpha \log |\Ab_\alpha|$.  The problem is that in $\rho_M$ we have not yet computed the partial trace, so the entropy of $\Ab_\alpha$ is also contributing.  There are various ways of dealing with this, I will adopt an approach from \cite{ohya2004quantum} based on introducing a modified version of the trace.  

For a generic von Neumann algebra $M$ on $\Hh$, the trace of a minimal projection is usually not one.  For example if $M$ is a factor, then any minimal projection is of the form $|v\ran\lan v|_A \otimes I_{\Ab}$, so its trace is $|\Ab|$.  We would like the entropy of this state on $\HA$ to be zero: ordinarily we would see this by computing the partial trace over $\Ab$ to obtain the pure state $|v\ran\lan v|$, but we would now like a way to see this that is intrinsic to $M$.  One natural way to do this is to define a new trace operation 
\be
\hat{\Tr}\equiv \frac{1}{|\Ab|}\Tr
\ee 
on $M$, which by construction has $\hat{\Tr} p=1$ for any minimal projection $p\in M$.  If we also define 
\be
\hat{\rho}_M\equiv |\Ab|\rho_M,
\ee
then for any $x$ in $M$ we have
\be
\mathbb{E}_\rho(x)=\Tr \rho x=\Tr\rho_M x=\hat{\Tr}\hat{\rho}_M x.
\ee
Moreover using eq. \eqref{rhoM2}, we see that 
\be
S(\rho,M)\equiv -\hat{\Tr}\hat{\rho}_M\log \hat{\rho}_M=-\Tr_A\rho_A\log \rho_A,
\ee
which thus gives an ``intrinsic'' definition of the entropy for the case of a factor: it is the expectation value in the state $\rho_M$ of the operator $-\log\hat{\rho}_M$, where recall that $\hat{\rho}_M$ is defined so that expectation values are computed using the modified trace $\hat{\Tr}$, which itself was defined to assign unit trace to minimal projections.    

We can extend these definitions to the case where $M$ is not a factor, by again introducing a modified trace $\hat{\Tr}$ which assigns unit trace to any minimal projection.  When $M$ is not a factor however, $\hat{\Tr}$ will not simply be proportional to $\Tr$.  This may be surprising, since often the trace is defined up to a constant factor as the unique linear operation on $\LH$ such that $\Tr xy=\Tr yx$.  We have more options here since we are only interested in defining $\hat{\Tr}$ on elements of $M$: there are no elements of $M$ that mix between different blocks in the decomposition \eqref{hilbertapp}, so we can normalize the trace independently in each block without disrupting the fact that $\hat{\Tr} xy=\hat{\Tr} yx$ for all $x,y\in M$.  This then enables us to define $\hat{\Tr}$ as the unique linear operation on $M$ which obeys  $\hat{\Tr} xy=\hat{\Tr} yx$ for all $x,y\in M$, and which gives $\hat{\Tr} p=1$ for any minimal projection $p\in M$.  In terms of the decompositions \eqref{hilbertapp}, \eqref{opapp}, if 
\be
x=\oplus_\alpha \left(x_\alpha\otimes I_{\Ab_\alpha}\right)
\ee
then we have
\be\label{hatt}
\hat{\Tr} x=\sum_\alpha \hat{\Tr}_\alpha \left(x_\alpha\otimes I_{\Ab_\alpha}\right)=\sum_\alpha \Tr_{A_\alpha} x_\alpha.
\ee
Similarly given any $\rho_M$ we can now also introduce a $\hat{\rho}_M$, which again is defined so that for any $x\in M$ we have
\be
\mathbb{E}_\rho(x)=\Tr \rho x=\Tr\rho_M x=\hat{\Tr}\hat{\rho}_M x.
\ee
Explicitly, given the expression \eqref{rhoM3} we then have
\be\label{rhoM4}
\hat{\rho}_M=\oplus_\alpha \left(p_\alpha \rho_{A_\alpha}\otimes I_{\Ab_\alpha}\right).
\ee
Finally we can define the entropy of the state $\rho$ on the algebra $M$ as
\be
S(\rho,M)\equiv -\hat{\Tr}\hat{\rho}_M\log \hat{\rho}_M,
\ee
which using \eqref{hatt} and \eqref{rhoM4} is easily shown to be equivalent to \eqref{entropy}.

\subsubsection{Properties of algebraic entropy}
We've now given a definition of the entropy of a state $\rho$ on an algebra:
\be\label{entropy2}
S(\rho,M)\equiv  -\hat{\Tr}\hat{\rho}_M\log \hat{\rho}_M = -\sum_\alpha p_\alpha \log p_\alpha +\sum_{\alpha}p_\alpha S(\rho_{A_\alpha}).
\ee
We see that the entropy has two parts: a ``classical'' piece given by the Shannon entropy of the probability distribution $p_\alpha$ for the center $Z_M$, and a ``quantum'' piece given by the average of the von Neumann entropy of each block over this distribution.  This entropy has several nice properties, which follow without too much difficulty from analogous properties of the ordinary von Neumann entropy:
\bi
\item $S(\rho,M)$ is invariant under $\rho\to u\rho u^\dagger$ for any unitary $u\in M$.  
\item $S(\rho,M)\geq 0$, with equality if and only if $\rho_M$ is a minimal projection.
\item $S(\rho,M)\leq \log \left(\hat{\Tr} I\right)= \log\left(\sum_\alpha|A_\alpha|\right)$, with equality if and only if $\rho_{A_\alpha}=\frac{I_{A_\alpha}}{|A_\alpha|}$ and $p_\alpha=\frac{|A_\alpha|}{\sum_\beta|A_\beta|}$.
\item $S\left(\sum_i \lambda_i\rho_i\right)\geq \sum_i \lambda_i S\left(\rho_i\right)$, where $\rho_i$ are any set of states and $\lambda_i\in[0,1]$ obey $\sum_i\lambda_i=1$.
\item If $\rho$ is pure, then $S(\rho,M)=S(\rho,M')$.  
\ei
We can also define the \textit{relative entropy} of two states $\rho$, $\sigma$ on $M$ as
\begin{align}\nonumber
S(\rho|\sigma,M)&\equiv \hat{\Tr}\left(\hat{\rho}_M\log \hat{\rho}_M-\hat{\rho}_M\log \hat{\sigma}_M\right)\\\nonumber
&=-S(\rho,M)+\mathbb{E}_\rho\left(-\log \hat{\sigma}_M\right)\\\label{rent}
&=\sum_\alpha p_\alpha^{\{\rho\}} \log \frac{p_\alpha^{\{\rho\}}}{p^{\{\sigma\}}_\alpha}+\sum_\alpha p_\alpha^{\{\rho\}}S\left(\rho_{A_\alpha}|\sigma_{A_\alpha}\right).
\end{align}

Again there is a ``classical'' contribution, measuring the distinguishability of the distributions $p_\alpha^{\{\rho\}}$ and $p^{\{\sigma\}}_\alpha$ on the center $Z_M$, and a ``quantum'' piece that averages the quantum relative entropy of each block over $p^{\{\rho\}}_\alpha$.  As with the usual relative entropy, we have $S(\rho|\sigma,M)\geq0$, with equality if and only if $\rho_M=\sigma_M$.

Finally if we define a modular Hamiltonian $\hat{K}^\rho_M \equiv -\log \hat{\rho}_M$, then the relative entropy is related to the ordinary entropy via 
\be\label{orents}
S(\rho|\sigma,M)=-S(\rho,M)+\hat{\Tr} \hat{\rho}_M \hat{K}^\sigma_M.
\ee

\bibliographystyle{jhep}
\bibliography{bibliography}
\end{document}